\documentclass{ieeeaccess}

\usepackage{cite}
\usepackage{amsmath,amssymb,amsfonts,amsthm}
\usepackage{algorithmic}
\usepackage{graphicx}
\usepackage{textcomp}
\def\BibTeX{{\rm B\kern-.05em{\sc i\kern-.025em b}\kern-.08em
    T\kern-.1667em\lower.7ex\hbox{E}\kern-.125emX}}
    
\usepackage{xcolor}
\usepackage{bm}

\usepackage{url}
\usepackage{cite}
\usepackage{float}

\usepackage{caption}
\usepackage{comment}
\usepackage[ruled,linesnumbered]{algorithm2e}

\usepackage{float}
\usepackage{multicol}
\usepackage{color}

\newcounter{multifig}
\usepackage{thmtools,thm-restate}
\usepackage{multirow}
\usepackage{mathtools, cuted}

\usepackage{lipsum}
\usepackage{array}
\usepackage{makecell}

\usepackage[caption=false, font=footnotesize]{subfig}

\newcommand\norm[1]{\left\lVert #1 \right\rVert}

\newtheorem{theorem}{Theorem}[section]

\newtheorem{definition}[theorem]{Definition}

\newtheorem{remark}[theorem]{Remark}

\newtheorem{problem}[theorem]{Problem}

\begin{document}
\history{Date of publication xxxx 00, 0000, date of current version xxxx 00, 0000.}
\doi{10.1109/ACCESS.2023.3283503}

\title{A Communication-efficient Local Differentially Private Algorithm in Federated Optimization }

\author{\uppercase{Syed Eqbal Alam}\authorrefmark{1}, \IEEEmembership{Member, IEEE},
\uppercase{Dhirendra Shukla\authorrefmark{1}, and Shrisha Rao}.\authorrefmark{2},
\IEEEmembership{Senior Member, IEEE}}
\address[1]{J. Herbert Smith Centre for Technology Management \& Entrepreneurship, Faculty of Engineering, University of New Brunswick, Fredericton, New Brunswick, Canada}
\address[2]{International Institute of Information Technology,
Bangalore, Karnataka, India}
\tfootnote{``Syed Eqbal acknowledges a partial support by Mitacs's grant number IT24468.
''}
\markboth
{Alam \headeretal: LDP-AIMD: Local Differential Privacy in Federated Optimization}
{Alam \headeretal: LDP-AIMD: Local Differential Privacy in Federated Optimization}

\corresp{Corresponding author: Syed Eqbal Alam (e-mail: syed.eqbal@unb.ca).}

\begin{abstract}
Federated optimization, wherein several agents in a network collaborate with a central server to achieve optimal social cost over the network with no requirement for exchanging information among agents, has attracted significant interest from the research community. In this context, agents demand resources based on their local computation. Due to the exchange of optimization parameters such as states, constraints, or objective functions with a central server, an adversary may infer sensitive information of agents.  { We develop a differentially-private additive-increase and multiplicative-decrease algorithm to allocate multiple divisible shared heterogeneous resources to agents in a network}. The developed algorithm provides a differential privacy guarantee to each agent in the network. { The algorithm does not require inter-agent communication, and the agents do not need to share their cost function or their derivatives with other agents or a central server}; however, they share their allocation states with a central server that keeps track of the aggregate consumption of resources. { The algorithm incurs very little communication overhead; for $m$ heterogeneous resources in the system, the asymptotic upper bound on the communication complexity is $\mathcal{O}(m)$ bits at a time step. Furthermore, if the algorithm converges in $K$ time steps, then the upper bound communication complexity will be $\mathcal{O}(mK)$ bits.}
{ The algorithm can find applications in several areas, including smart cities, smart energy systems, resource management in the sixth generation ($6$G) wireless networks with privacy guarantees, etc. 
} We present experimental results to check the efficacy of the algorithm. Furthermore, we present empirical analyses for the trade-off between privacy and algorithm efficiency.
\end{abstract}

\begin{IEEEkeywords}{ Additive increase multiplicative decrease algorithm}, AIMD algorithm, Differential privacy, Federated optimization, {Multi}-resource allocation, { heterogeneous resources}, {Multi-agent system, Communication-efficient resource allocation, Optimization and control}.
\end{IEEEkeywords}

\tfootnote{{\bf Citation Information:} S. E. Alam, D. Shukla and S. Rao, "A Communication-Efficient Local Differentially Private Algorithm in Federated Optimization," in IEEE Access, vol. 11, pp. 58254-58268, 2023, doi: 10.1109/ACCESS.2023.3283503.}

\titlepgskip=-15pt

\maketitle

\section{Introduction}
As the number of mobile and IoT devices and the volume of data increase, the risk of privacy leaks also increases. Moreover, because of privacy concerns, users may not wish to share their on-device data, which could be detrimental to many applications' usability.
Federated optimization is a step toward resolving this issue. In federated optimization, resource-constrained clients (agents) coordinate with a central server to solve an optimization problem without exposing their private data~\cite{Reddi2021, KonecnyMR2015}. Federated optimization has been used in federated learning that facilitates clients to collaboratively learn models without sharing their private on-device data. Instead, clients share the learned parameters with a central server. The central server aggregates the updates sent by the clients to update the global model~\cite{Mcmahan2017}. More information is available in recent comprehensive surveys~\cite{Kairouz2021, Yin2021}.

Although federated optimization provides privacy to agents, stronger privacy guarantees can be obtained by combining it with privacy techniques ~\cite{Mcmahan2017} such as  differential privacy \cite{Koskela2020, MAbadi2016, FiorettoH2019}, { federated $f$-differential privacy that guarantees a differential privacy to agents against an individual adversary and a group of adversaries~\cite{Zheng2021}}, { secure multiparty computation~\cite{Cramer2015}---a cryptography-based technique, etc.} Briefly, differential privacy provides a certain privacy guarantee to users storing data in a shared database~\cite{Dwork2006, Dwork2014}.  It allows an analyst to learn aggregate statistics on the database without compromising the privacy of participating users. This is done by adding randomness to the data via a trusted curator, following a certain probability distribution. Local differential privacy is proposed for scenarios where the curator is untrusted and may act as an adversary. In local differential privacy, users randomize their data before storing it in the database~\cite{Kairouz2016, Ding2017}. 
%

%
In federated optimization, agents make decisions based on their local on-device computations. Because they exchange certain optimization parameters, such as states, constraints, objective functions, derivatives of objective functions, or multipliers, with the central server, an adversary may infer sensitive information of agents~\cite{Han2017, Olivier2020}. Furthermore, communication overhead between clients/agents and the central server is another major issue in federated optimization~\cite{ Caldas2019, He2022, Hamer2020}. { Privacy-preserving federated optimization solutions have applications in many areas; a few are listed as follows, healthcare \cite{Rieke2020}}, smart cities {\cite{ZZheng2022}}, the Internet of things \cite{Briggs2021}, {\cite{Zhang2023}, supply chain management---collective supply chain risk prediction \cite{Zheng2023}}, smart energy systems \cite{Dobbe2020}. { Interested readers can refer to the recent survey for more applications} \cite{Yin2021}. { In addition, because of the massive growth of resource-constrained IoT devices and the requirement of low-latency communications, scalability, and the quality of user experiences in the sixth generation ($6$G) wireless networks~\cite{YANG2022}, communication-efficient privacy-preserving federated optimization solutions will have many promising applications along with resource management. 
%
} 

{ Most of the work on federated optimization with differential privacy involves machine learning and deep learning techniques~\cite{Koskela2020, MAbadi2016, FiorettoH2019, Zheng2021}; we could not find any work on federated optimization with differential privacy for allocating multiple shared heterogeneous resources. This paper fills this gap.}
We develop a local differential privacy-additive increase multiplicative decrease, called \emph{LDP-AIMD} algorithm, {for multi-resource allocation in} federated settings that provides privacy guarantees to agents in the network with little communication overhead. { Additionally, the agents in the network achieve close to optimal allocations, and the network achieves close to a social optimal value}. { We consider multivariate cost functions coupled through the allocation of multiple heterogeneous resources; it is more challenging to solve such optimization problems than single variable problems; furthermore, the single resource allocation solutions used for multi-resource settings are not efficient, and provide sub-optimal solutions~\cite{Ghodsi2011}. Finally, the developed algorithm provides privacy guarantees to the agents' partial derivatives of cost functions.
%
}
The LDP-AIMD algorithm is motivated by the ideas from the additive increase and multiplicative decrease (AIMD) algorithm~\cite{Chiu1989, Syed2018_al}. 
The AIMD algorithm was proposed for congestion control in the Transmission Control Protocol (TCP) and is widely used in resource allocation. AIMD has recently been used to solve optimization problems wherein there is no requirement for inter-agent communication, but agents coordinate with a central server to solve the problem~\cite{Wirth2014, Syed2018_ISC2}. The algorithm incurs little communication overhead over the network.  However, there is no prior work on AIMD with differential privacy for federated optimization.

The LDP-AIMD algorithm consists of two phases: additive increase and multiplicative decrease.  In the additive increase phase, agents continually increase their resource demands until they receive a one-bit capacity constraint notification from a central server.  The central server broadcasts the one-bit capacity event notification when the aggregate resource consumption reaches its capacity.  After receiving the notification, agents enter into the multiplicative decrease phase and decrease their resource demands.  They calculate their resource demands in the multiplicative decrease phase based on the average allocation, noisy derivatives of the cost function, and other parameters.  Noise is added to the derivatives of the cost function drawn from a certain probability distribution.  In line with prior work~\cite{Wirth2014, Corless2016, Boyd2011}, we consider that the cost function of an agent is convex, continuously differentiable, and increasing in each variable. They are coupled through the allocation of multiple divisible heterogeneous resources. 
LDP-AIMD provides a differential privacy guarantee to the partial derivatives of the agents' cost functions and obtains close-to-optimal values for multiple resources. Additionally, the solution incurs very little communication overhead. { For $m$ heterogeneous divisible resources in the system, the asymptotic upper bound communication complexity will be $\mathcal{O}(m)$ bits at a time step. Furthermore, if the algorithm converges in $K$ time steps, then the upper bound communication complexity will be $\mathcal{O}(mK)$ bits.}

AIMD algorithm and its variants are implemented in several real-world applications such as electric vehicle (EV) charging \cite{Zishan2021, Shah2019}, sharing economy \cite{Syed2018_B}, optimal power generation in micro-grids \cite{Crisostomi2014}, decentralized power sharing \cite{Fan2021}, collaborative cruise control systems \cite{Wirth2014}, other applications can be found at \cite{Corless2016}, \cite{Crisostomi2018}.
Additionally, federated optimization, learning, and local differential privacy are used in many real-world applications ~\cite{Cormode2018, JWang2021}. For example, Google deployed a local differentially private model called {\em RAPPOR}~\cite{Erlingsson2014} to collect data on statistics of a user's Chrome web browser usage and settings~\cite{Chromium2022} without compromising the privacy of participating users. Microsoft deployed local differential privacy mechanisms to collect telemetry users' data to enhance users' experiences with privacy guarantees to the users~\cite{Ding2017}. Apple deployed local differential privacy techniques for their servers to learn new words generated by the users' devices that are not in the device dictionary; they classify these words in groups with privacy guarantees to the users~\cite{Thakurta2017}.
Furthermore, Google uses federated learning in Gboard (Google's virtual keyboard) for next-word prediction~\cite{AndrewH2018}, and emoji prediction from the typed text~\cite{Swaroop2019}. NVIDIA and Massachusetts General Brigham Hospital developed a federated learning-based model to predict COVID-19 patients' oxygen needs during the recent pandemic~\cite{NVIDIA2020}.



\subsection{Paper Contributions}
{ As discussed previously, though there are several works on federated optimization with differential privacy, most of the work deal with training a global model in the machine learning application areas (federated learning). Our work, however, is the first in allocating multiple heterogeneous divisible resources in federated settings with differential privacy guarantees, which brings many opportunities and challenges.} Our main contributions to this paper are as follows:
\begin{itemize}

\item[$\square$] We model the multi-resource allocation problem as a federated optimization problem that provides differential privacy guarantees to agents in the network. { In the model, each agent has a multivariate cost function wherein a variable represents the allocation of a resource. Such optimization problems are more challenging to solve than single-variable optimization problems.

\item[$\square$] Our algorithm is based on the additive increase multiplicative decrease (AIMD) algorithm, incurs little communication overhead, provides close to optimal solutions and guarantees differential privacy.}
{ Existing AIMD-based resource allocation solutions do not automatically provide differential privacy guarantees to participating agents~\cite{Wirth2014, Zishan2021, Shah2019, Syed2018_ISC2,Alam2022_phd}.}

In the model, agents do not need to share their private information with other agents or a server. 
The agents randomize their partial derivatives of the cost functions before calculating their resource demands.
In the model, we consider a central server that keeps track of aggregate resource consumption at a time instant and broadcasts a one-bit notification in the network when aggregate resource consumption reaches resource capacity. Although the central server knows the aggregate resource consumption, it does not know the partial derivatives of the cost functions or the cost functions of the agents.
%
Furthermore, the solution incurs negligible communication overhead; the server broadcasts a one-bit capacity constraint notification in the network when aggregate resource demands reach resource capacity. Additionally, the communication complexity is independent of the number of agents in the network. { For $m$ resources in the system, the asymptotic upper bound communication complexity will be $\mathcal{O}(m)$ bits at a time step. Furthermore, if the algorithm converges in $K$ time steps, then the upper bound communication complexity will be $\mathcal{O}(mK)$ bits.}

\item[$\square$] { We show theoretical results for the LDP-AIMD algorithm for noises drawn from Gaussian and Laplacian distributions.}   
\item[$\square$] We present experimental results to check the empirical effectiveness of the algorithm. We observe that the algorithm provides $(\epsilon_{1i} + \epsilon_{2i})$-local differential privacy and $(\epsilon_{1i}+\epsilon_{2i}, \delta_{1i} +\delta_{2i})$-local differential privacy guarantees to an agent for two shared resources in the system, respectively, depending on the probability distribution of noise. We also analyze the trade-off between privacy and the algorithm's efficiency. 
\end{itemize}
{ \subsection{Outline}
The paper is organized as follows: Section \ref{lab:notations} presents the notations used in the paper and a brief background on the AIMD algorithm and deterministic AIMD algorithm. Section \ref{prob_form_1} presents the problem formulation. Section \ref{lab:LDP-AIMD} presents our developed algorithm, the LDP-AIMD algorithm. Additionally, theoretical results on the LDP-AIMD algorithm with noise drawn from the Laplacian and Gaussian distributions are presented in the section. Section \ref{exp_results} presents the experimental setup and results. Section \ref{sec:background} describes the related work on distributed resource allocation, distributed optimization, federated optimization, and differential privacy. Section \ref{conc} provides the conclusion of the paper, and finally, Appendix  \ref{proofs} presents the proofs of theorems stated in Section \ref{lab:LDP-AIMD}. }

\section{Notations and Background} \label{lab:notations}
\subsection{Notations}
Let us consider that $n$ agents in a multi-agent system collaborate to access $m$ shared divisible heterogeneous  resources and wish to minimize the social cost of the system. Let the capacities of resources be $C_1, C_2, \ldots, C_m$, respectively. Let the set of real numbers be denoted by $\mathbb{R}$, the set of positive real numbers be denoted by $\mathbb{R}_+$, and the set of $m$-dimensional vectors of real numbers be denoted by $\mathbb{R}^m$. For $i=1,2,\ldots,n$ and $j=1,2,\ldots,m$, let the amount of resource $j$ allocated to agent $i$ be denoted by $x_{ji} \in [0,C_j]$. Furthermore, suppose that allocating resources incurs a certain cost to each agent, captured in the agent's cost functions. Let the cost function of agent $i$ be denoted by $f_i: \mathbb{R}^m_+ \to \mathbb R_+$. Let the cost function $f_i$ be twice continuously differentiable, strictly convex, and increasing in each variable. 
Let $\nu=0,1,2,\ldots$, denote the time steps, and let $x_{ji}(\nu)$ denote the amount of instantaneous allocation of resource $j$ of agent $i$ at time step $\nu$. 

For $n$ agents in the network, we model the multi-resource allocation problem as the following federated optimization problem:

\begin{problem} \label{obj_fn1}
	\begin{align*} 
		\begin{split}
			\min_{{x}_{11}, \ldots, {x}_{mn}} \quad &\sum_{i=1}^{n} f_i(x_{1i}, x_{2i},
			\ldots, x_{mi}),    		\\
			\mbox{subject to} \quad
			&\sum_{i=1}^{n} x_{1i} = C_1; \sum_{i=1}^{n} x_{2i} = C_2; \ldots; \sum_{i=1}^{n} x_{mi} = C_m		\\
			&x_{11} \geq 0; x_{12} \geq 0; \ldots; x_{mn} \geq 0.
		\end{split}
	\end{align*}
\end{problem}
Let the solution to this optimization problem be denoted by $x_{ji}^*$. We assume that the solution is strictly positive. Furthermore, let $x^{*} = (x_{11}^{*}, \ldots, x_{mn}^{*})$. As the cost function $f_i$ is strictly convex and the constraint sets are compact, there exists a unique optimal solution \cite{Wirth2014}. { We assume that all agents have strictly convex cost functions that are differentiable and increasing. Furthermore, the agents are cooperative and demand the resources with their true valuations with the aim to minimize the overall cost to the network.}

\subsection{AIMD Algorithm}
We present the fundamental additive-increase and multiplicative-decrease (AIMD) algorithm. The AIMD algorithm consists of two phases---additive increase and multiplicative decrease. In the additive increase (AI) phase, an agent increases its resource demand linearly by a constant $\alpha_j \in (0,C_j]$, called {\em additive increase factor}, until it receives a one-bit {\em capacity event} notification from the central server. The central server tracks the aggregate resource consumption and broadcasts the one-bit capacity event notification when the total resource consumption reaches its capacity. For time steps $\nu=0,1,2,\ldots$, the additive increase phase is formulated as
\begin{align}  \label{updateDAIMD_add}
	x_{ji}(\nu+1) = x_{ji} (\nu) + \alpha_j, \ j=1,2,\ldots,m.
\end{align}
After receiving the capacity event notification, an agent responds in a probabilistic way to decrease its resource demand. Suppose that the $k$'th capacity event occurs at time step $\nu$. Then, for $0 \leq \beta_j <1$ and $\bm{\beta_j} \in \{\beta_j,1\}$ with certain probability distribution, the multiplicative decrease (MD) phase is formulated as
\begin{equation}
      x_{ji} (\nu+1) = \bm{\beta_j} x_{ji}(\nu).
\end{equation}
After the multiplicative decrease phase, agents again enter into the AI phase until they receive the next capacity event notification, and they continue this. Interested readers can refer to \cite{Wirth2014} and \cite{Corless2016} for details on the AIMD and the AIMD-based resource allocation models.

\subsection{Deterministic AIMD Algorithm}   \label{algo_det}
We briefly describe here the deterministic AIMD algorithm by \cite{Syed2018_al}, which is the starting point of the LDP-AIMD 
algorithm. The deterministic AIMD algorithm consists of two phases: The additive increase and the multiplicative decrease phases. In the additive increase phase, an agent increases its resource demand linearly by the additive increase factor $\alpha_j \in (0,C_j]$, until it receives a one-bit {\em capacity event} notification from the central server when the aggregate demand of a resource reaches its capacity. For time step $\nu=0,1,2,\ldots$, the additive increase phase is formulated as
\begin{align}  \label{updateDAIMD_add1}
	x_{ji} (\nu+1) = x_{ji} (\nu) + \alpha_j.
\end{align}
After receiving the capacity event notification, agents decrease their resource demands in a deterministic way using multiplicative decrease factor $0 \leq \beta_j <1$ and scaling factor $0 < \lambda_{ji} \leq 1$.
 Let $k_j$ denote the capacity event of resource $j$, for $j=1,2,\ldots,m$. Let $\overline{x}_{ji}(k_j)$ denote the amount of average allocation of resource $j$ at capacity event $k_j$ for agent $i$. For all $i$, $j$, and $k_j$, the average allocation is calculated as follows:
\begin{align}  \label{average_eqn}
	\overline{x}_{ji}(k_j)=\frac{1}{k_j+1} \sum_{\ell=0}^{k_j} x_{ji}(\ell).
\end{align}

Suppose that the $k_j$'th capacity event occurs at time step $\nu$. Then the multiplicative decrease phase is formulated as
\begin{align}   \label{updateDAIMD_mul1}
x_{ji} (\nu+1) = \Big (\lambda_{ji}(k_j) \beta_j  + \big(1 -\lambda_{ji}(k_j) \big) \Big)x_{ji}(\nu). 
\end{align} 
The agent $i$ calculates $\lambda_{ji}(k_j)$ for resource $j$ as in \eqref{prob_x}.
Moreover, after decreasing the demands, agents again start increasing their demands linearly by $\alpha_j$ until they receive the next capacity event notification. This process repeats over time to obtain optimal values over long-term average allocations. 

For $i=1,2,\ldots,n$, $j=1,2,\ldots,m$, let $\frac{\partial}{\partial_{x_{ji}}} f_i(x_{1i}, x_{2i}, \ldots, x_{mi})$ denote the partial derivative with respect to $x_{ji}$ of cost function $f_i(\cdot)$. 
For $k \in \mathbb{N}$, let $t_{k_j}$ denote the time at which $k_j$'th capacity event occurs. Let $\overline{x}_{1i}(t_{k_1})$ denote the average allocation of agent $i$ over the capacity events for resource $1$ until time instant $t_{k_1}$, and $\overline{x}_{2i}(t_{k_1})$ denote the average allocation of agent $i$ over capacity events for resource $2$ until time instant $t_{k_1}$.
Recall that $\overline{x}_{ji}(t_{k_j})$ is also denoted by $\overline{x}_{ji}(k_j)$; it is calculated as in \eqref{average_eqn}, for $j=1,2,\ldots,m$. 
Additionally, let $\Gamma_j>0$ denote the normalization factor of resource $j$, for $j=1,2,\ldots,m$. It is chosen such that $0<\lambda_{ji}(k_j) \leq 1$ holds.
For all $i$, $j$, at the $k_j$'th capacity event $\lambda_{ji}(k_j)$ is obtained as follows:
\begin{align} \label{prob_x} 
	\lambda_{ji}(k_j) = \Gamma_j  
	\frac{\frac{\partial}{\partial_{x}} \Big \vert_{x = \overline{x}_{ji}(t_{k_j})} f_i(\overline{x}_{1i}(t_{k_j}), \ldots,
		\overline{x}_{mi}(t_{k_j}))}{\overline{x}_{ji}(t_{k_j})}.
\end{align}
%
%
%
Furthermore, each agent runs its algorithm to demand shared resources. It is experimentally shown  that following the approach with appropriate values of $\lambda_{ji}(k_j)$, $\alpha_j$, $\beta_j$, and $\Gamma_j$, the long-term average allocations of agents converge to optimal values, that is, 
\begin{align*}
\lim_{k_j \to \infty} \overline{x}_{ji}(k_j) = x_{ji}^{*},
\end{align*}
 and the system achieves an optimal social cost
 
  $\sum_{i=1}^n f_i(x_{1i}^{*}, \ldots,x_{mi}^{*})$.	

%

In the deterministic AIMD algorithm, an adversary may obtain the actual allocations of an agent and may infer the derivatives of the cost function or the cost function of the agent. On the other hand, the LDP-AIMD algorithm provides privacy guarantees to agents in the network on the partial derivatives of their cost functions and obtains close to optimal social cost.

\section{Problem formulation} \label{prob_form_1}

%

%
%

We aim to protect the privacy of agents' partial derivatives of cost functions with a certain privacy guarantee. Let us assume that agent $i$ stores its private information in data-set $D_{i}$, for $i=1,2,\ldots,n$.
%
%
Let $k_1$ capacity events of resource $1$ and $k_2$ capacity events of resource $2$ occur until time instant $t_k$. Recall that, for $k \in \mathbb{N}$, $\overline{x}_{1i}(t_{k})$ denote the average allocation of agent $i$ over the capacity events for resource $1$ until time instant $t_{k}$ ( refer Equation \eqref{average_eqn}). And, $\overline{x}_{2i}(t_{k})$ denotes the average allocation of agent $i$ over capacity events for resource $2$ until time instant $t_{k}$.
Note that for the sake of simplicity of notations, we consider two resources here. 

For $i=1,2,\ldots,n$,

 let $f^o_{1,i}(t_{k}) \triangleq \frac{\partial}{\partial_{x}} \Big \vert_{x = \overline{x}_{1i}(t_{k})} f_{i}(\overline{x}_{1i}(t_{k}), \overline{x}_{2i}(t_{k}))$ and $f^o_{2,i}(t_{k}) \triangleq \frac{\partial}{\partial_{x}} \Big \vert_{x = \overline{x}_{2i}(t_{k})} f_{i}(\overline{x}_{1i}(t_{k}), \overline{x}_{2i}(t_{k}))$. 
For $i=1,2,\ldots,n$ and $\overline{x}'_{1i}, \overline{x}'_{2i} \geq 0$,

 let $f'^o_{1,i}(t_{k}) \triangleq \frac{\partial}{\partial_{x}} \Big \vert_{x = \overline{x}'_{1i}(t_{k})} f_{i}(\overline{x}'_{1i}(t_{k}), \overline{x}'_{2i}(t_{k}))$ and $f'^o_{2,i}(t_{k}) \triangleq \frac{\partial}{\partial_{x}} \Big \vert_{x = \overline{x}'_{2i}(t_{k})} f_{i}(\overline{x}'_{1i}(t_{k}), \overline{x}'_{2i}(t_{k}))$. 
Let $D_{1,i}(t_{k})$ denote the set of partial derivatives of the cost functions of agent $i$ up to the $k_1$'th capacity event of resource $1$ (occurred until time instant $t_k$), we define it as $D_{1,i}(t_{k}) \triangleq \left \{f^o_{1,i}(t_{0}), f^o_{1,i}(t_{1}),\ldots, f^o_{1,i}(t_{k}) \right \}$.
Additionally, let $D_{2,i}(t_{k})$ denote the set of partial derivatives of the cost functions of agent $i$ up to the $k_2$'th capacity event of resource $2$ (occurred until time instant $t_k$), we define it as $D_{2,i}(t_{k}) \triangleq \left \{f^o_{2,i}(t_{0}), f^o_{2,i}(t_{1}),\ldots, f^o_{2,i}(t_{k}) \right \}$. 
 
Let $\Delta q_{1i}$ denote the sensitivity for resource $1$, and let $\Delta q_{2i}$ denote the sensitivity for resource $2$. We define the $p$-norm sensitivity as follows.
\begin{definition}[$p$-norm Sensitivity] \label{def:sensitivity}
For $i=1,2,\ldots,n$, let $f^o_{1,i} = \frac{\partial}{\partial_{x}} \Big \vert_{x = \overline{x}_{1i}} f_{i}(\overline{x}_{1i}, \overline{x}_{2i})$ and $f^o_{2,i} = \frac{\partial}{\partial_{x}} \Big \vert_{x = \overline{x}_{2i}} f_{i}(\overline{x}_{1i}, \overline{x}_{2i})$,
and for $\overline{x}'_{1i}, \overline{x}'_{2i} \geq 0$, let $f'^o_{1,i} = \frac{\partial}{\partial_{x}} \Big \vert_{x = \overline{x}'_{1i}} f_{i}(\overline{x}'_{1i}, \overline{x}'_{2i})$ and $f'^o_{2,i} = \frac{\partial}{\partial_{x}} \Big \vert_{x = \overline{x}'_{2i}} f_{i}(\overline{x}'_{1i}, \overline{x}'_{2i})$. 
 For $p \in \mathbb{N}$, the $p$-norm sensitivity metric for agent $i$ for resource $1$ is defined as
\begin{align} \label{eq:sensitivity}
		\Delta q_{1i} \triangleq \norm{f^o_{1,i} - f'^o_{1,i}}_p. 
	\end{align}
	Analogously, the $p$-norm sensitivity metric $\Delta q_{2i}$ for resource $2$ is defined.
\end{definition}
%
For agent $i$, $i=1,2, \ldots,n$, we define the local $\epsilon_i$-differential privacy as follows:
\begin{definition}[Local $\epsilon_i$-differential privacy] \label{def:Local-DP} 
Let $f_{i}: \mathbb{R}_+^2 \to \mathbb{R}_+$ be the cost functions of agent $i$. 
%
%
Let the set $\mathcal{D}_{i}$ be the input values to the privacy mechanism, defined as
\begin{align*}
\mathcal{D}_{i} & \triangleq \{ \frac{\partial}{\partial{x}_{1i}} f_{i}(x_{1i}, x_{2i}), \frac{\partial}{\partial{x}_{2i}} f_{i}(x_{1i}, x_{2i})\} \\ & \cup \{ 
\frac{\partial}{\partial{x}'_{1i}} f_{i}(x'_{1i}, x'_{2i}), \frac{\partial}{\partial{x}'_{2i}} f_{i}(x'_{1i}, x'_{2i})  \}.
\end{align*}  Let $S$ be the sample space. Furthermore, let $M_{q_i}:S \times \mathcal{D}_{i} \to \mathbb{R}$ be a privacy mechanism, and $q_i: \mathcal{D}_{i} \to \mathbb{R}$ be the query on $\mathcal{D}_{i}$. Then for $\epsilon_i > 0$ and output values $\eta_i \in \mathbb{R}$, if the following holds
	\begin{align}
	&\mathbb{P} \left(M_{q_i} \left (\frac{\partial}{\partial{x}_{1i}} f_{i}(x_{1i}, x_{2i}) \right) =\eta_i \right) \nonumber \\ &\leq \exp{(\epsilon_i)} \cdot \mathbb{P} \left(M_{q_i} \left (\frac{\partial}{\partial{x}'_{1i}} f_{i}(x'_{1i}, x'_{2i}) \right) =\eta_i \right),
	\end{align}
	then $M_{q_i}$ is called an $\epsilon_i$-local differential privacy mechanism.
\end{definition}
Here, $\epsilon_i$ is the privacy loss bound of agent $i$.

\begin{definition}[Local $(\epsilon_i, \delta_i)$-differential privacy] \label{def:Local-DP2} 
Let $f_{i}: \mathbb{R}_+^2 \to \mathbb{R}_+$ be the cost functions of agent $i$. 
%
%
Let the set $\mathcal{D}_{i}$ be the input values to the mechanism, defined as
\begin{align*}
\mathcal{D}_{i} &\triangleq \{ \frac{\partial}{\partial{x}_{1i}} f_{i}(x_{1i}, x_{2i}), \frac{\partial}{\partial{x}_{2i}} f_{i}(x_{1i}, x_{2i})\} \\&\cup \{\frac{\partial}{\partial{x}'_{1i}} f_{i}(x'_{1i}, x'_{2i}), \frac{\partial}{\partial{x}'_{2i}} f_{i}(x'_{1i}, x'_{2i}) \}.
\end{align*}
 Let $S$ be the sample space. Furthermore, let $M_{q_i}:S \times \mathcal{D}_{i} \to \mathbb{R}$ be a privacy mechanism, and $q_i: \mathcal{D}_{i} \to \mathbb{R}$ be the query on $\mathcal{D}_{i}$.
%
%
%
Then for $\epsilon_i, \delta_i > 0$, and for all input values and for all output values $\eta_i \in \mathbb{R}$, if the following holds
	\begin{align}
	& \mathbb{P} \left(M_{q_i} \left (\frac{\partial}{\partial{x}_{1i}} f_{i}(x_{1i}, x_{2i}) \right) =\eta_i \right) \nonumber \\ & \leq \exp{(\epsilon_i)} \cdot \mathbb{P} \left(M_{q_i} \left (\frac{\partial}{\partial{x}'_{1i}} f_{i}(x'_{1i}, x'_{2i}) \right) =\eta_i \right) + \delta_i,
	\end{align}
	then $M_{q_i}$ is called an $(\epsilon_i, \delta_i)$-local differential privacy mechanism.
\end{definition}
Here, $\epsilon_i$ and $\delta_i$ are the {\em privacy loss} bounds; $\delta_i$ denotes the probability that the output of the mechanism $M_{q_{i}}$ varies by a multiplicative factor of $\exp(\epsilon_i)$ when applied to  data-sets $\mathcal{D}_{i}$. The smaller values of $\delta_i$ and $\epsilon_i$ imply that higher privacy is preserved. Note that for $\delta_i=0$, the mechanism preserves $\epsilon_i$-local differential privacy. 

Recall that the average allocation is calculated as in \eqref{average_eqn}; let the vector $\overline{x} \in (\mathbb{R}_+^{n})^m$ denote $(\overline{x}_{11}, \ldots, \overline{x}_{mn})$. Also recall that the solution to federated optimization Problem \ref{obj_fn1} is denoted by $x^{*} = (x_{11}^{*}, \ldots, x_{mn}^{*})$. In this paper, we develop a local differentially private AIMD algorithm, the {\em LDP-AIMD} algorithm. The LDP-AIMD solves the federated optimization Problem \ref{obj_fn1} with close to optimal values and provides a privacy guarantee on the partial derivatives of the cost functions of agents. That is, for two resources in the system, $j=1,2$ and $k_j \in \mathbb{N}$, we obtain
\begin{align*}
	\lim_{k_j \to \infty} \overline{x}_{ji}(k_j) \approx x_{ji}^*, \ \mathrm{i=1,2,\ldots,n}, 
\end{align*}
with a (local) differential privacy guarantee on the partial derivatives of the cost function of each agent with negligible communication overhead.

\section{The LDP-AIMD Algorithm} \label{lab:LDP-AIMD}
This section presents the local differentially private additive-increase multiplicative-decrease (LDP-AIMD) algorithm. {  As we discussed earlier, in the deterministic AIMD algorithm (described in Section \ref{algo_det}), adversaries may infer the actual resource allocations of agents and may also infer the partial derivatives of the cost functions of an agent. Although the agents make their decisions in the deterministic AIMD algorithm based on their local computations; however, the agents are prone to privacy breach.} The LDP-AIMD algorithm bridges the gap and provides a strong privacy guarantee on the partial derivatives of agents' cost functions and so on their cost functions and obtains close to optimal social cost.

\subsection{LDP-AIMD algorithm for multiple resources}

\begin{figure*}[!ht]
    \centering
\includegraphics[width=1\textwidth,clip=true,trim=7.5cm 8.05cm 0cm 4.8cm]{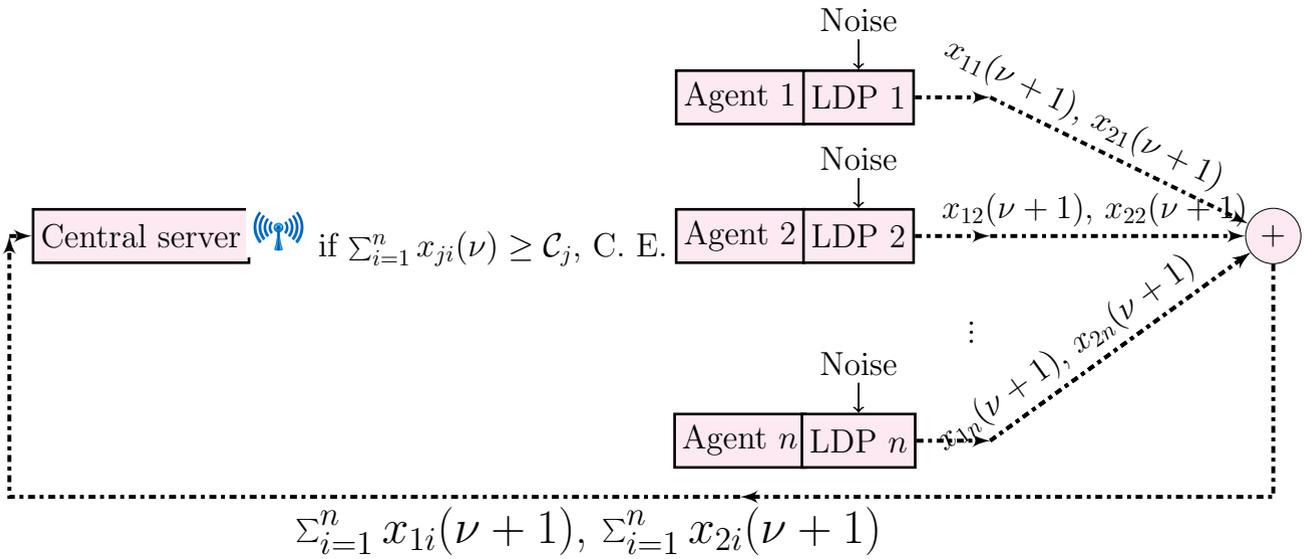}
	\caption{The local differential privacy AIMD's (LDP-AIMD) system diagram for multi-resource allocation. Here, C.E. represents the broadcast of a one-bit capacity event notification in the network.}
		\label{fig:Sys_diag}
\end{figure*}

In the LDP-AIMD algorithm, each agent works as a curator and randomizes partial derivatives of its cost function. The LDP-AIMD algorithm has two phases similar to the classical AIMD algorithm---additive increase (AI) and multiplicative decrease (MD). We denote the additive increase factor by $\alpha_j>0$, multiplicative decrease factor by $0 \leq \beta_j<1$, and normalization factor by $\Gamma_j>0$, for $j=1,2,\ldots,m$. Each agent runs its LDP-AIMD algorithm to demand shared resources. Additionally, we consider a central server that keeps track of the aggregate allocation of resources and broadcasts a one-bit capacity event notification when the aggregate demand reaches the capacity of the resource. The central server initializes the parameters $\Gamma_j$, $\alpha_j$, and $\beta_j$ with desired values and sends these to agents when they join the network. The algorithm works as follows: After joining the network, agents start increasing their demands for shared resource $j$ linearly by the additive increase factor $\alpha_j>0$ until they receive a one-bit capacity event notification from the central server. The central server broadcasts this notification when the aggregate resource demand reaches resource capacity.
We illustrate the system diagram in Figure \ref{fig:Sys_diag}.

For discrete time steps $\nu=0,1,2,\ldots,$ we describe the AI phase as:
\begin{align}  \label{eq:AI-updated}
	x_{ji} (\nu+1) = x_{ji} (\nu) + \alpha_j.
\end{align}
After receiving the capacity event notification from the central server, an agent decreases its demands using multiplicative decrease factor $0 \leq \beta_j<1$, the average resource allocations, and the noisy scaling factor $\hat{\lambda}_{ji}$. The noisy scaling factor $\hat{\lambda}_{ji}(k_j)$ is calculated using the partial derivatives of cost functions $f_i$ of agent $i$ and  noise drawn from a certain probability distribution. Let $d_{ji}(k_j)$ be a random variable and denotes the noise added to the partial derivatives of cost functions of agent $i$ for resource $j$. It is drawn from certain probability distribution at the $k_j$'th capacity event. 
At the $k_j$'th capacity event, for $i=1,2,\ldots,n$ and $j=1,2,\ldots,m$, we obtain $\hat{\lambda}_{ji}(k_j)$ as follows:
\begin{align}  \label{prob_x2}
	&\hat{\lambda}_{ji}(k_j) \nonumber \\&= \Gamma_j  
		\frac{\Big \vert\frac{\partial}{\partial_{x}} \Big \vert_{x = \overline{x}_{ji}(k_j)} f_i(\overline{x}_{1i}(t_{k_j}), \ldots, \overline{x}_{mi}(t_{k_j})) + d_{ji}(k_j) \Big \vert}{\overline{x}_{ji}(k_j)}.
\end{align}
%
The value of normalization factor $\Gamma_j$ is chosen such that $ 0< \hat{\lambda}_{ji}(k_j) \leq 1$.
%
%

%

%
%
Let us suppose that the $k_j$'th capacity event occurs at time step $\nu$. Additionally, let $S_j(\nu) \in \{0,1\}$ denote the capacity event notifications, for resource $j$, $j=1,2,\ldots,m$; it is updated to $S_j(k_j)=1$, when the $k_j$'th capacity event occurs.
The multiplicative decrease phase is formulated as follows:
\begin{align}  \label{eq:noisyMD}
	x_{ji} (\nu+1) = \Big (\hat{\lambda}_{ji}(k_j) \beta_j  + \big(1 -\hat{\lambda}_{ji}(k_j) \big) \Big)x_{ji}(\nu). 
\end{align} 
%
%

After decreasing the resource demands, the agents again start to increase their demands linearly until they receive the next capacity event notification; this process is repeated over time. 
By following this, agents achieve privacy guarantees on their partial derivatives, and their long-term average allocations reach close to optimal allocations. 
%
%
The differentially private AIMD algorithm for agents is presented in Algorithm \ref{algo_DP}, and the algorithm of the central server is presented in Algorithm \ref{algo_DP_server}.
\begin{algorithm}[htb] 
\LinesNumbered
Input: $\sum_{i=1}^n x_{ji}(\nu)$, for $j = 1,2,\ldots,m$, $\nu \in \mathbb{N}$\;
	Output:
	$S_{j}(\nu) \in \{0,1\}$, for $j = 1,2,\ldots,m$, $\nu \in
	\mathbb{N}$\;
	Initialization: $\Gamma_j \leftarrow \frac{1}{1000}$\;
	\For{$\nu=0,1,2,\ldots$}{
		
		\For{$j = 1,2,\ldots,m$}{				
			
			\eIf{$\sum_{\ell=0}^{n} x_{ji}(\nu) \geq C_j$}{
			%
			Update the capacity event notification:
				\begin{align*}
					S_{j}(\nu+1) \leftarrow 1;
				\end{align*}
			}{
			\begin{align*}
					S_{j}(\nu+1) \leftarrow 0;
				\end{align*}
			}
	} }
	
	\caption{Algorithm of the central server.}
	\label{algo_DP_server}
\end{algorithm}
The LDP-AIMD algorithm provides privacy guarantees to agents, and the average allocations asymptotically reach close to the optimal values, and the system incurs negligible communication overhead. Note that the more noise is added to the partial derivatives of the cost functions, the more privacy is guaranteed, but the convergence of average allocations will go farther from the optimal point, and the solution's efficiency will decrease.

{ Readers can find a similar formulation for the expected value of opinions of agents used in \cite{YShang2018} in the context of social dynamical systems.}

\begin{algorithm}[htb]   Input: $\Gamma_j$, $S_j$, for $j = 1,2,\ldots,m$\;
	Output:
	$x_{ji}(k+1)$, for $j = 1,2,\ldots,m$, $k \in
	\mathbb{N}$\;
	
	Initialization: $x_{ji}(0) \leftarrow 0$ and
	$\overline{x}_{ji}(0) \leftarrow x_{ji}(0)$, $k_j \leftarrow 0$, for
	$j = 1,2,\ldots,m$\;
	
	\For{$\nu=0,1,2,\ldots$}{
		
		\For{$j = 1,2,\ldots,m$}{				
			
			\eIf{$S_{j}(\nu)=1$}{
			$k_j \leftarrow k_j+1$\;
			Update the scaling factor $\hat{\lambda}_{ji}(k_j)$ as in \eqref{prob_x2}\;
			Update the resource demand:
				\begin{align*}
					&x_{ji}(\nu+1) \\ \leftarrow &\left (\hat{\lambda}_{ji}(k_j) \beta_j   +\left(1-\hat{\lambda}_{ji}(k_j)\right) \right) x_{ji}(\nu);
				\end{align*}
			}{
				\begin{align*}
					x_{ji}(\nu+1) \leftarrow x_{ji}(\nu) + \alpha_j;
				\end{align*}
			}
			
	} }
	
	\caption{Local differentially private AIMD (LDP-AIMD) algorithm of agent $i$.}
	\label{algo_DP}
\end{algorithm}

{ We state the following asymptotic upper bound for the communication complexity of the algorithm.
\begin{remark}[Communication complexity]
     For a multi-agent system with $m$ resources, the communication complexity in the worst-case scenario will be $\mathcal{O}(m)$ bits at a time step. Moreover, if the algorithm converges in $K$ time steps, then the communication complexity will be $\mathcal{O}(mK)$ bits.
 \end{remark}}

  Let $\Delta q_{ji}$ denote the sensitivity for partial derivatives of the cost functions of agent $i$ for resource $j$. 
We now present local differential privacy results with noises drawn from the Laplacian and the Gaussian distributions to calculate the scaling factors defined in~\eqref{prob_x2}.
\subsection{Noise drawn from Laplacian distribution}
In this subsection, we obtain the result for Laplace distribution. The agent $i$ adds the noise $d_{ji}(t_k)$ to its partial derivatives of the cost functions to calculate its scaling factors (see \eqref{prob_x2}). The noise is drawn from the Laplacian distribution with probability density function $p_{ji}(z) = \frac{\exp \left(-|z|/(\Delta q_{ji}/\epsilon_{ji}) \right)}{2\Delta q_{ji}/\epsilon_{ji}}$ with location $0$ and scale parameter $\Delta q_{ji}/\epsilon_{ji}$, for $\epsilon_{ji} > 0$. For two resources in the system, agent $i$ obtains $(\epsilon_{1i}+\epsilon_{2i})$-differential privacy guarantee, as in the following result.

\begin{theorem} \label{thm:Laplacian_loc}
Let $ \mathcal{D}_{i}$ be the data universe of the partial derivatives of cost function $f_i$ of agent $i$.
Let $M_{q_{1i}}$ and $M_{q_{2i}}$ denote the privacy mechanisms for resource $1$ and $2$, respectively, of agent $i$. Let $\mathcal{A}_i$ denote the privacy mechanism for coupled resources. Moreover, for two resources in the multi-agent system, let $\eta_{1i}, \eta_{2i} \in \mathbb{R}$ denote the output and $\epsilon_{1i}, \epsilon_{2i} > 0$ denote the privacy loss bounds of the privacy mechanisms $M_{q_{1i}}$ and $M_{q_{2i}}$, respectively.
Additionally, let $d_{1i}$ and $d_{2i}$ represent noise drawn from Laplace distribution with location $0$ and variance $2 (\Delta q_{ji}/\epsilon_{ji})^2$ used in the privacy mechanisms, then the coupled privacy mechanism $\mathcal{A}_i$ is $(\epsilon_{1i} + \epsilon_{2i})$-differentially private.
\end{theorem} 
\begin{proof}
It is presented in the Appendix.
\end{proof}

\subsection{Noise drawn from Gaussian distribution}
This subsection presents the results for Gaussian noise added to an agent's partial derivative of its cost function to calculate its scaling factor (see \eqref{prob_x2}). If the noise $d_{ji}(t_k)$ is drawn from the Gaussian distribution with probability density function $p_{ji}(z) = \frac{1}{\sqrt{2\pi \sigma_{ji}}} \exp{(-\frac{z^2}{2 \sigma_{ji}^2})}$ with zero-mean and standard deviation $\sigma_{ji} \geq \frac{\Delta q_{ji}}{\epsilon_{ji}} \sqrt{2 (\ln \frac{1.25}{\delta_{ji}})}$; agent $i$ obtains $(\epsilon_{1i}+\epsilon_{2i}, \delta_{1i}+\delta_{2i})$-differential privacy guarantee, as we obtain in the following result.
\begin{theorem} \label{thm:Gaussian_loc}
Let $\mathcal{D}_{i}$ be data-sets of the partial derivatives of cost function $f_i$ of agent $i$.
Let $M_{q_{1i}}$ and $M_{q_{2i}}$ denote the privacy mechanisms for resource $1$ and $2$, respectively, of agent $i$. Let $\mathcal{A}_i$ denote the privacy mechanism for coupled resources. Moreover, for two resources in the multi-agent system, let $\eta_{1i}, \eta_{2i} \in \mathbb{R}$ denote the output and $\epsilon_{1i}, \epsilon_{2i} > 0$ denote the privacy loss bounds of the privacy mechanisms $M_{q_{1i}}$ and $M_{q_{2i}}$, respectively.
Additionally, let $d_{1i}$ and $d_{2i}$ be noise drawn from Gaussian distribution with mean zero and standard deviation $\sigma_{1i}$ and $\sigma_{2i}$, respectively, used in the privacy mechanisms, then the coupled privacy mechanism $\mathcal{A}_i$ is $(\epsilon_{1i}+\epsilon_{2i},\delta_{1i} + \delta_{2i})$-differentially private.
\end{theorem} 

\begin{proof}
The proof sketch is presented in the Appendix.
\end{proof}

{ We make the following remark on adding the noise to the additive increase phase.
\begin{remark}
    We expect similar results when noise is added to the additive increase phase; however, we expect a slower convergence. Furthermore, as our algorithm aims to preserve the privacy of the partial derivatives of the cost functions of agents, we add noise in the multiplicative decrease phase only.
\end{remark}}

\section{Results} \label{exp_results}
In this section, we describe the experimental setup for the LDP-AIMD algorithm and present the results. The results show that with appropriate noise levels, the long-term average allocations reach close to the optimal values with a privacy guarantee on the derivatives of the cost functions of agents. Furthermore, we analyze the trade-off between privacy loss and the algorithm's accuracy---the more noise we add, the better the privacy, but the lesser the accuracy. 


\subsection{Setup}
We consider six agents (the number is a random pick) that share two resources with capacities $C_1=5$ and $C_2=6$. The agents initialize their allocations $x_{1i}(0)=x_{2i}(0)=0$, for $i=1,2,\ldots,6$. We chose the additive increase factors $\alpha_1 = 0.01$ and $\alpha_2 = 0.0125$, the multiplicative decrease factors $\beta_1 = 0.70$ and $\beta_2 = 0.6$. Additionally, we chose the normalization factors $\Gamma_1 = \Gamma_2 = \frac{1}{1000}$.
To create different cost functions, we use uniformly distributed integer random variables $a_i \in [10,30]$, $b_i \in [15,35]$, $h_i \in [10,20]$, and $g_i \in [15,25]$, for $i=1,2,\ldots,6$. 
The cost functions are listed  as follows:	
\begin{align} \label{cost_fn}
	& f_{i}(x_{1i}, x_{2i}) = 
	\begin{cases}
		(i) \frac{1}{2} a_i x_{1i}^2 + \frac{1}{4} b_i x_{1i}^4 +  \frac{1}{2} b_i x_{2i}^2  +  \frac{1}{4} a_i x_{2i}^4 \\
		(ii)  \frac{1}{2} b_i x_{1i}^2 +  \frac{1}{4} b_i x_{2i}^2 \\
		(iii) \frac{1}{2} b_i x_{1i}^4 + \frac{1}{3} b_i x_{2i}^4. \\
	\end{cases}
\end{align}
Agents $1$ and $2$ demand the resources using cost function as in \eqref{cost_fn}(i); analogously, agents $3$ and $4$ demand the resources using cost function as in \eqref{cost_fn}(ii), and agents $5$ and $6$ demand the resources using cost function as in \eqref{cost_fn}(iii). 

Recall that, for agent $i$, $\Delta q_{1i}$ and $\Delta q_{2i}$ represent the sensitivity of the partial derivatives of cost functions with respect to resource $1$ and resource $2$, respectively. 
For a fixed time step $t_k$, let $k_1$ capacity events of resource $1$ and $k_2$ capacity events of resource $2$ occur until $t_k$. We add the same level of noise for all agents sharing a resource; to do so, we calculate the maximum sensitivity until time step $t_k$ for resource $1$ as $\Delta q_{1}(t_k) = \max \{ \Delta q_{11}(t_k), \ldots, \Delta q_{1n}(t_k) \}$. Analogously, for resource $2$, as $\Delta q_{2} (t_k)= \max \{ \Delta q_{21}(t_k), \ldots, \Delta q_{2n}(t_k) \}$. 
%
\subsection{Main results} 
We now present experimental results to check the efficacy of the LDP-AIMD algorithm. Moreover, we present the results with Gaussian noise added to the partial derivatives of the cost functions of agents the agents obtain $(\epsilon_{1i}+\epsilon_{2i}, \delta_{1i}+\delta_{2i})$-differential privacy guarantee, and we also present the results with Laplacian noise added to the partial derivatives of the cost functions of agents, the agents obtain $(\epsilon_{1i}+\epsilon_{2i})$-differential privacy guarantee and agents' allocations on long-term averages are close to optimal values. 
%


%
\subsubsection{\bf{With Gaussian noise:}}
This subsection presents the results obtained by adding Gaussian noise with mean zero and standard deviation $\sigma_{ji}$ to the partial derivatives of the cost functions of agents, $j=1,2$. In this case, the sensitivity is calculated as the $l_2$-norm of the consecutive partial derivatives of cost functions of agents. 
Furthermore, the standard deviation for agent $i$ is calculated as $\sigma_{1i} = \frac{\Delta q_{1i}}{\epsilon_{1i}} \sqrt{2 (\ln \frac{1.25}{\delta_{1i}})}$ for resource $1$; similarly, standard deviation for agent $i$ for resource $2$ is calculated as $\sigma_{2i} = \frac{\Delta q_{2i}}{\epsilon_{2i}} \sqrt{2 (\ln \frac{1.25}{\delta_{2i}})}$.  For a fixed $\delta_{1i} = \delta_{1u} = \delta_1$ and $\sigma_{1i} = \sigma_{1u} = \sigma_{1}$, and the sensitivity $\Delta q_1 = \max \{ \Delta q_{11}, \ldots, \Delta q_{1n}\}$, all agents have the same privacy loss values, $\epsilon_{1i} = \epsilon_{1u} = \epsilon_1$, for $i,u=1,2,\ldots,6$. Similarly, for a fixed $\delta_{2i} = \delta_{2u} = \delta_2$ and $\sigma_{2i} = \sigma_{2u} = \sigma_{2}$, and the sensitivity $\Delta q_2 = \max \{ \Delta q_{21}, \ldots, \Delta q_{2n}\}$, we have, $\epsilon_{2i} = \epsilon_{2u}=\epsilon_2$, for $i,u = 1,2,\ldots,6$. In the experiment, the Gaussian noise with mean zero and standard deviation $\sigma_{1}$ is added to the partial derivatives of the cost functions with respect to resource $1$; analogously, Gaussian noise with mean zero and standard deviation $\sigma_{2}$ is added to the partial derivatives of the cost functions with respect to resource $2$. 
 For resource $1$, we chose $\delta_1 = 0.01$ with $l_2$-norm sensitivity $\Delta q_{1} = 1.32$ and standard deviation for Gaussian noise $\sigma_1 = 20.50$ to obtain the privacy loss value $\epsilon_1 = 0.2$. Likewise, for resource $2$, we chose $\delta_2 = 0.01$ with $l_2$-norm sensitivity $\Delta q_{2} = 2.53$ and standard deviation for Gaussian noise $\sigma_2 = 39.31$ to obtain the privacy loss value $\epsilon_2 = 0.2$. 
As the random variable (noise) $d_{ji}(k_j)$ follows Gaussian distribution and $\sigma_{ji}(k_j) = \frac{\Delta q_{ji}(k_j)}{\epsilon_{ji}} \sqrt{2 (\ln \frac{1.25}{\delta_{ji}})}$, the mechanism is $(\epsilon_{1i}+\epsilon_{2i}, \delta_{1i}+\delta_{2i})$-local differentially private (See Theorem 1, Section 4).
The initial few partial derivatives are not considered while calculating the maximum sensitivities because they depend on the initialization states. 
\begin{figure}[ht]
	\centering
	\subfloat[]{%
		\includegraphics[width=0.9\linewidth]{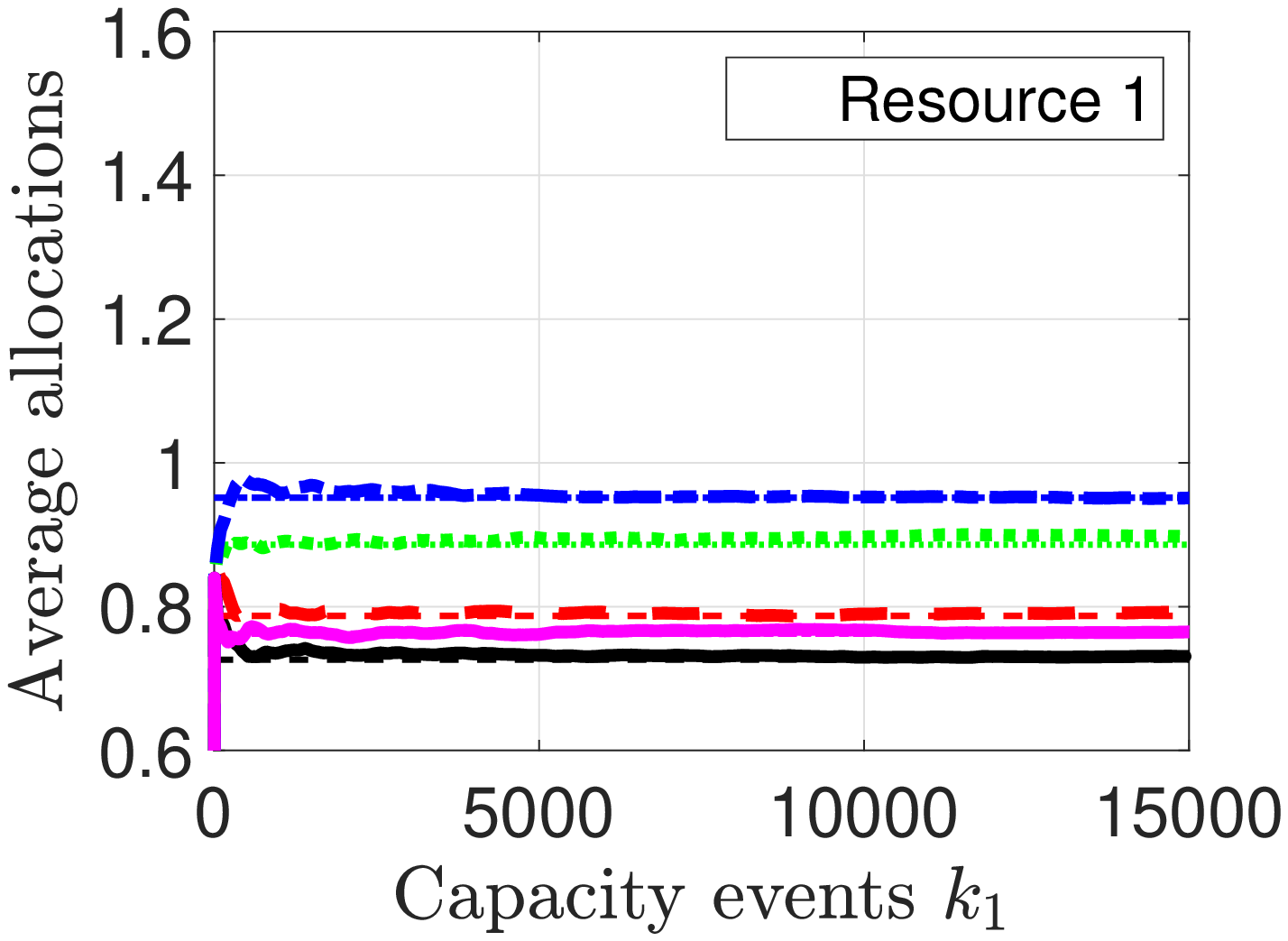}}		\label{avg_r1}\hfill
	\subfloat[]{%
		\includegraphics[width=0.9\linewidth]{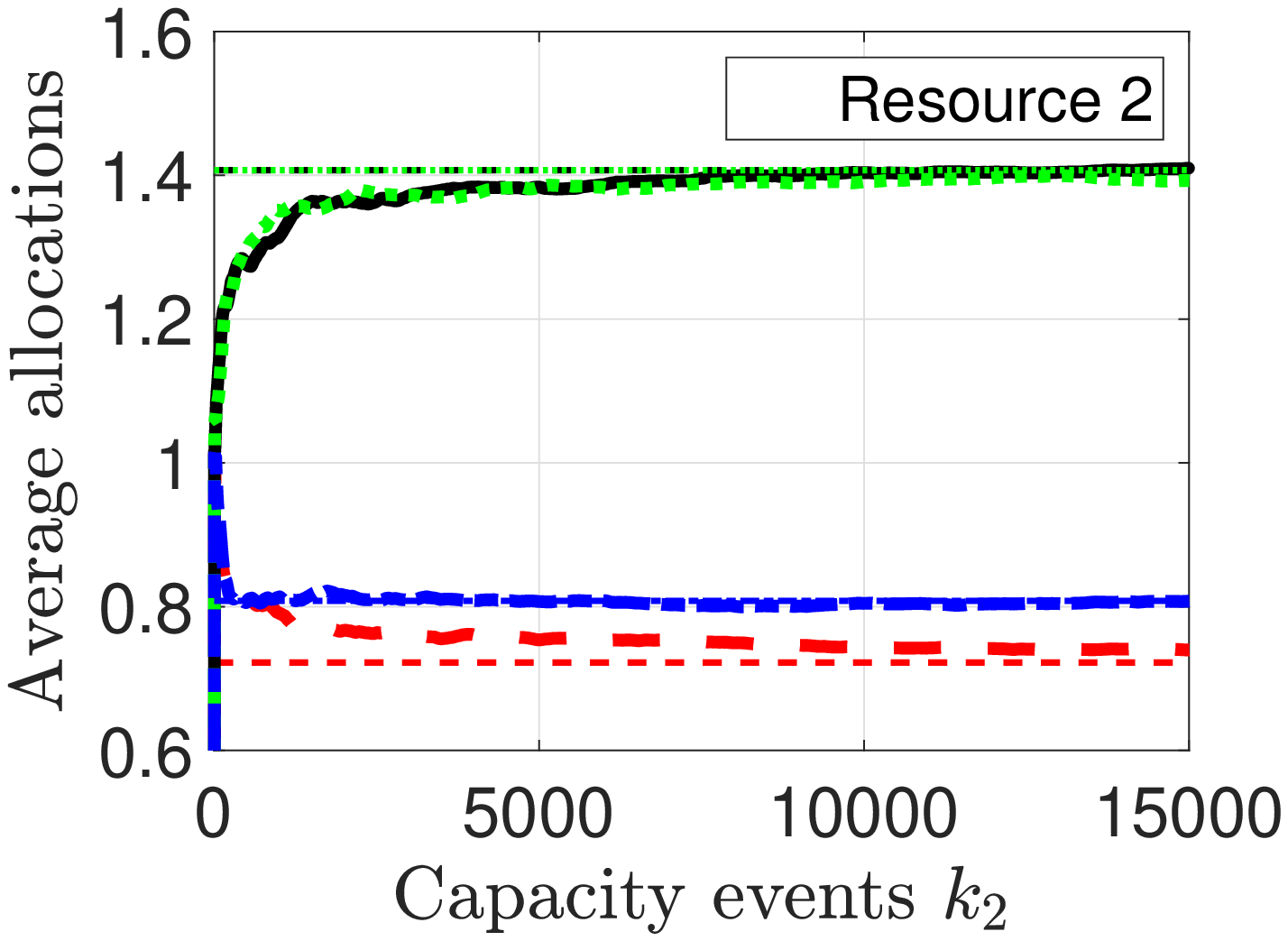}}
	\label{avg_r2} \hfill
\caption{Evolution of average allocations for selected agents (a) Gaussian noise with mean zero and standard deviation $20.5$ added to the partial derivatives of the cost functions, for resource $1$, (b) Gaussian noise with mean zero and standard deviation $39.31$ added to the partial derivatives of the cost functions for resource $2$. Dashed straight lines are plotted with the optimal values (obtained by the solver).}
	\label{fig1} 
\end{figure}

Figures \ref{fig1} and \ref{fig2} are based on Gaussian noise with mean zero and standard deviations $\sigma_1 = 20.50$ and $\sigma_2 = 39.31$, for resources $1$ and $2$, respectively; the noise is added to the partial derivatives of the cost functions of agents at each time step. Figure \ref{fig1} illustrates the evolution of average allocations of resources for a few selected agents. We observe in Figure \ref{fig1} that the average allocations of resources reach close to the optimal values over time with the chosen privacy losses and the sensitivities; the noise is drawn from Gaussian distribution. To compare the results for optimal values, we solved the optimization Problem using a solver; the optimal values by the solver are denoted by the dashed straight lines of the respective colors in Figure \ref{fig1}. 
 Figure \ref{fig2} presents the evolution of partial derivatives of the cost functions with respect to a resource as shaded error bars. We observe that partial derivatives gather closer to each other over time; that is, error values decrease over time. Thus, the partial derivatives of the cost function with respect to a resource make a consensus over time. Hence, the long-term average allocations are close to optimal values with chosen privacy losses and sensitivities.
\begin{figure}[!ht]
	\centering
	\subfloat[]{%
		\includegraphics[width=0.9\linewidth]{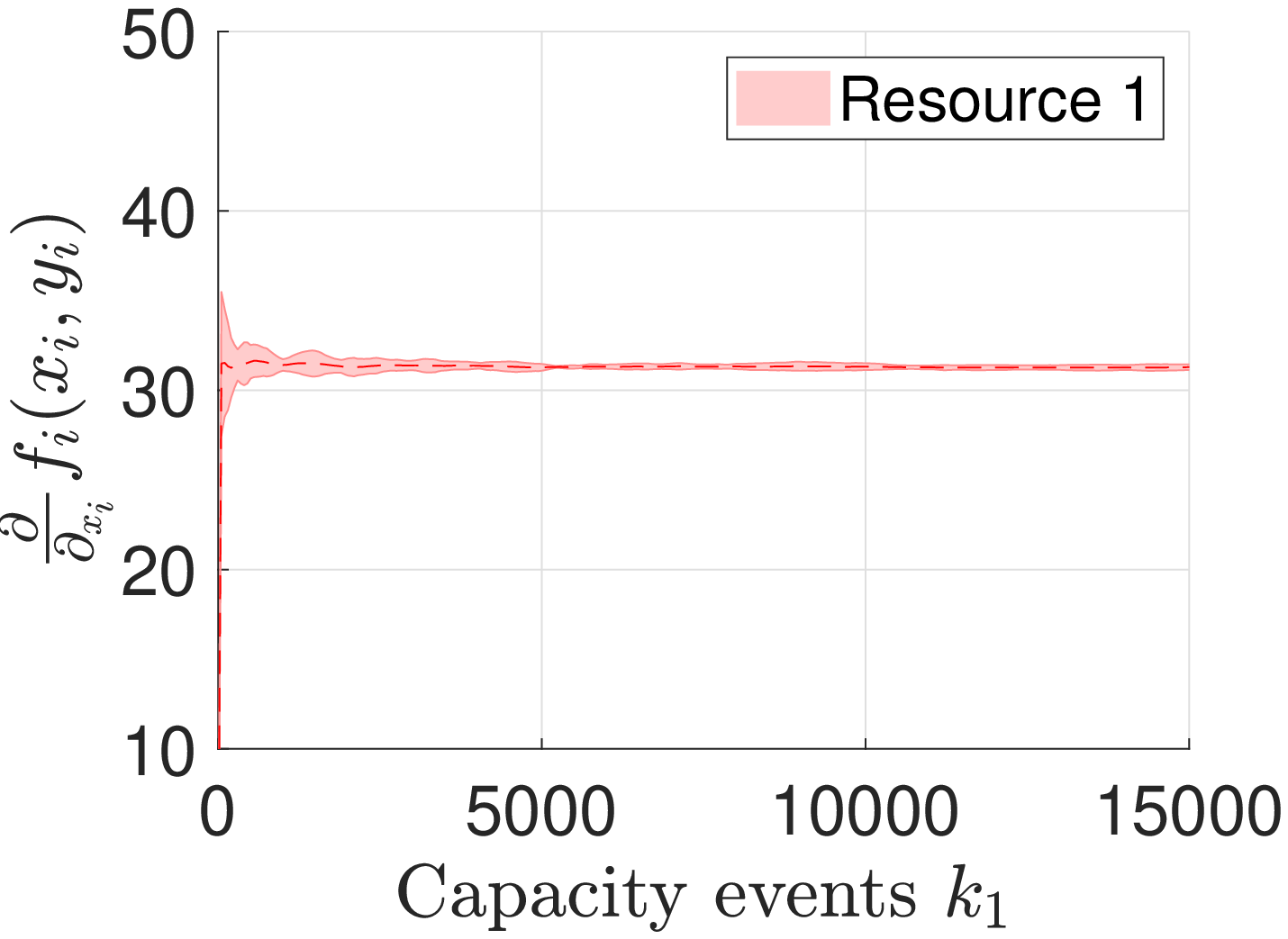}}
	\label{gradx}\hfill
	\subfloat[]{%
		\includegraphics[width=0.9\linewidth]{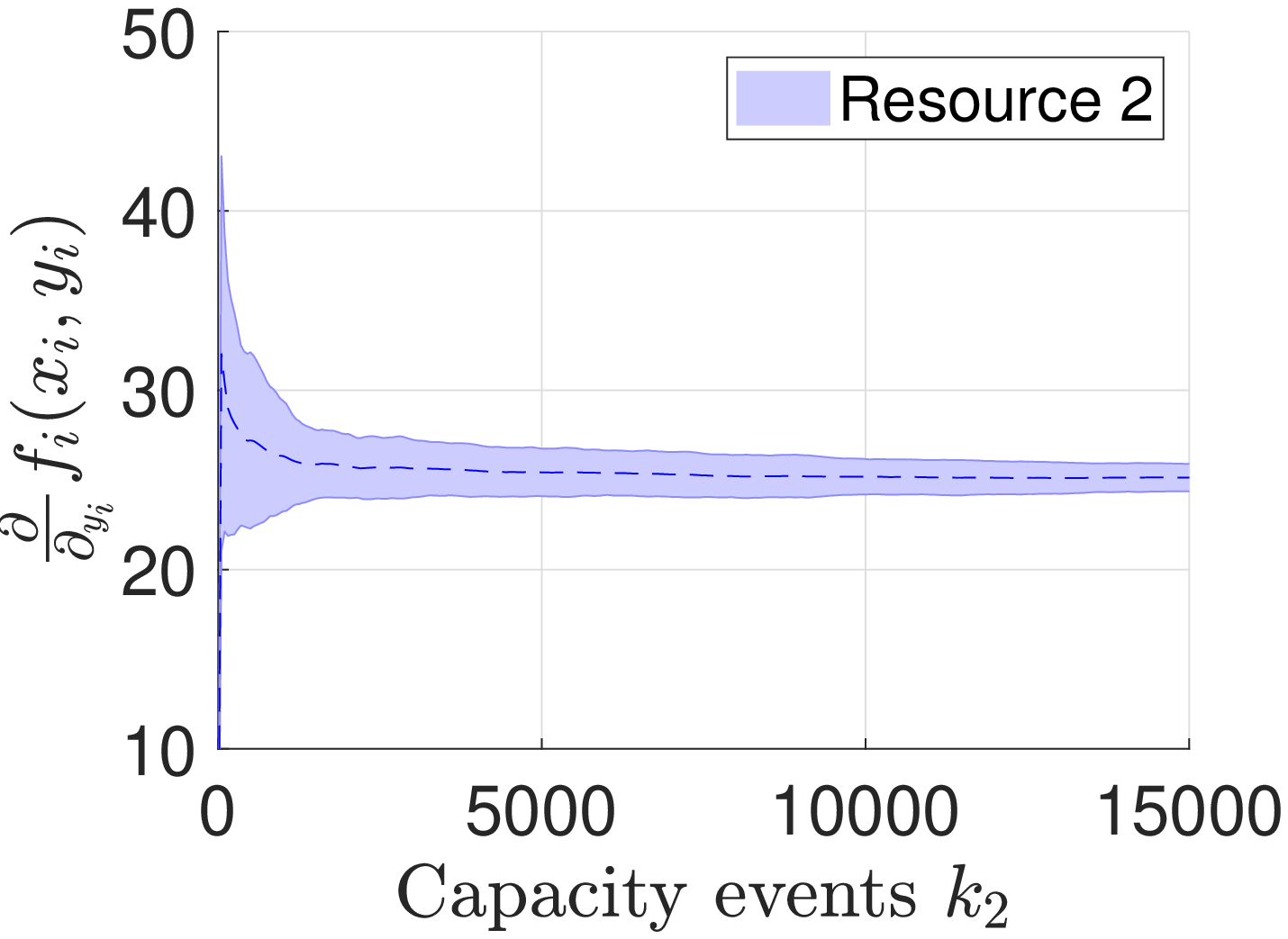}}
	\label{grady} \hfill
	
\caption{(a) Evolution of derivatives of cost functions with respect to resource $1$ of all agents, (b) evolution of derivatives of cost functions with respect to resource $2$ of all agents, for a single simulation. The noise is drawn from the Gaussian distribution.} \label{fig2} 
\end{figure}

\begin{figure}[!ht]
	\centering
	\subfloat[]{%
		\includegraphics[width=0.495\linewidth]{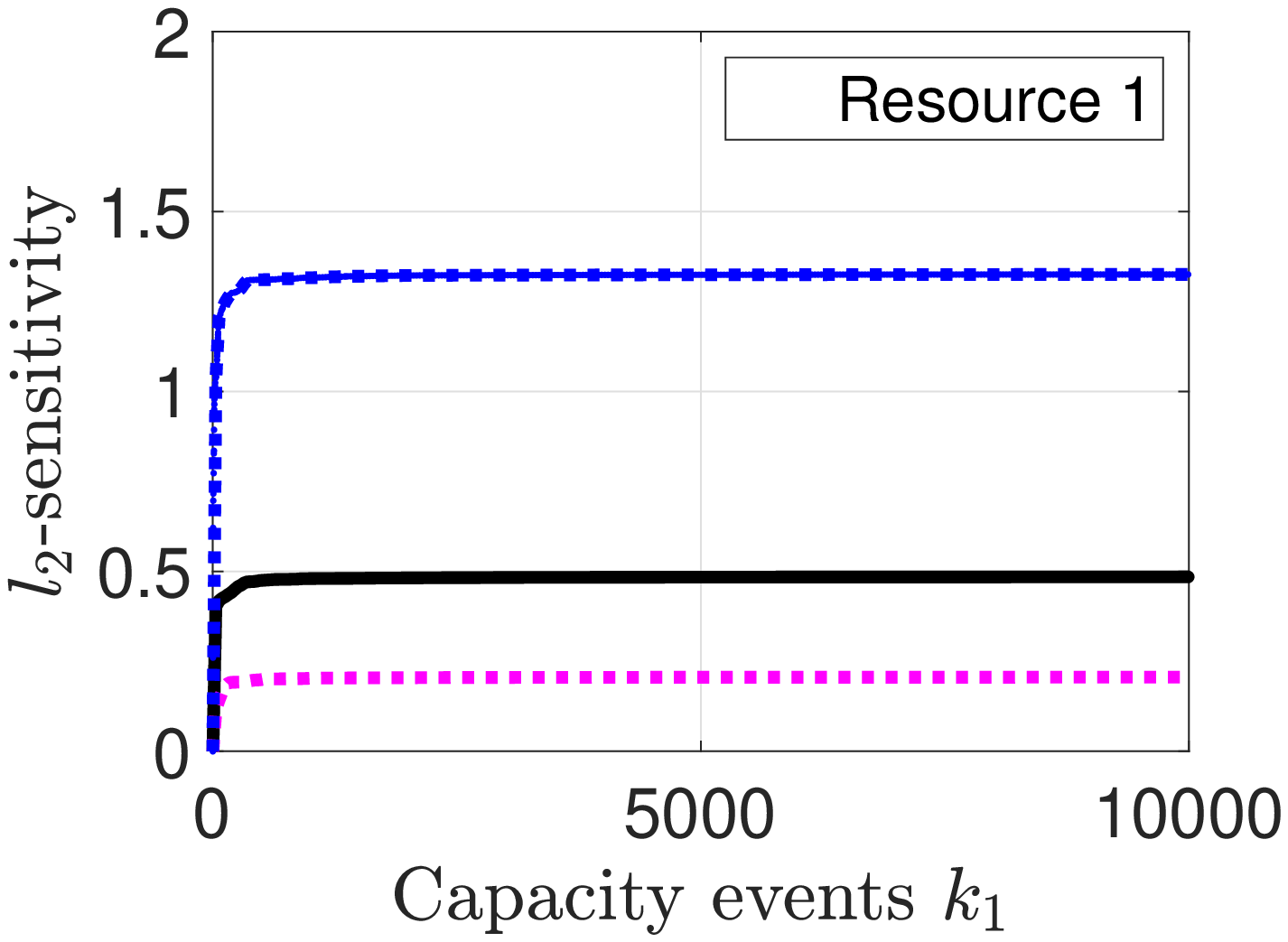}}
	\label{gradx}\hfill
	\subfloat[]{%
		\includegraphics[width=0.495\linewidth]{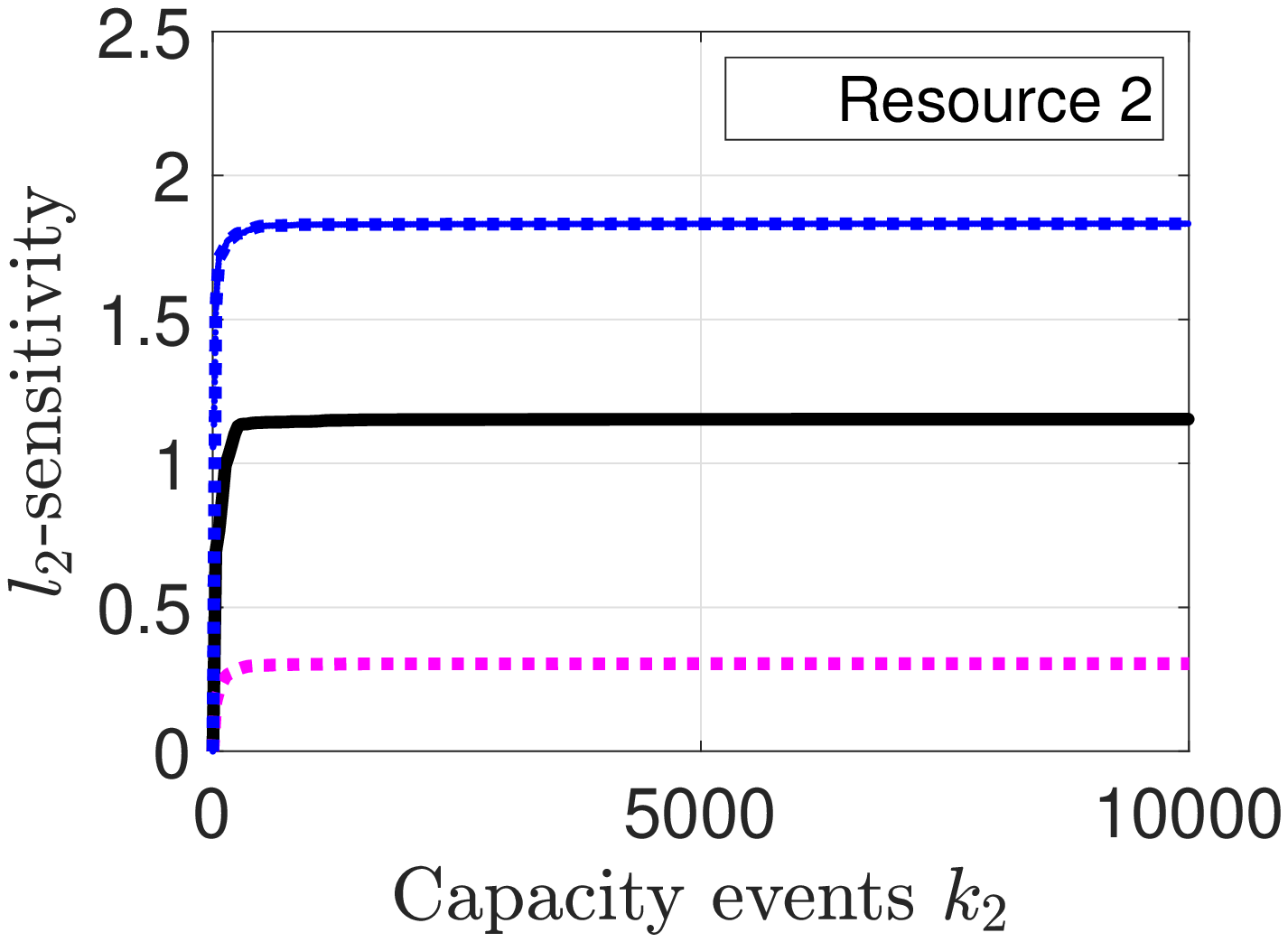}}
	\label{grady} \hfill
	
\caption{(a) Evolution of $l_2$-sensitivity of resource $1$, (b) evolution of $l_2$-sensitivity of resource $2$.} \label{fig:l2-sensiv} 
\end{figure}
\begin{figure}[!ht] 
	\centering
	\subfloat[]{%
		\includegraphics[width=1\linewidth]{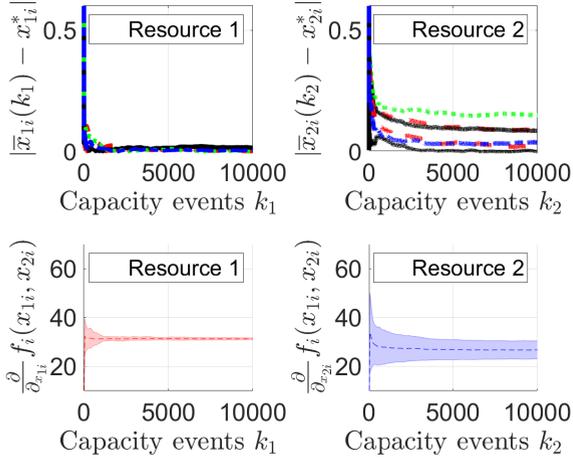}}
	\label{avg_gradsN144}\hfill
		\subfloat[]{%
		\includegraphics[width=1\linewidth]{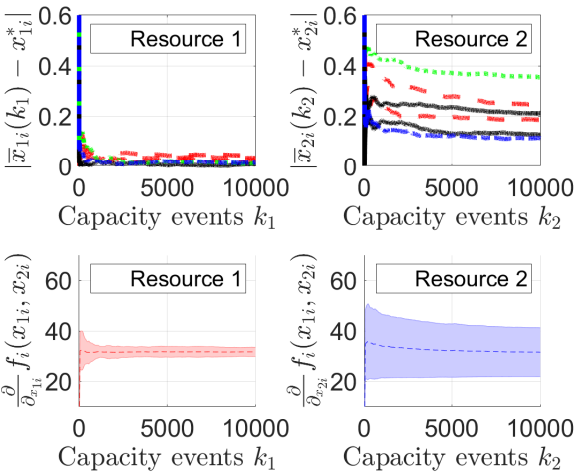}}
	\label{avg_gradsN720}\hfill
	\caption{ Evolution of absolute difference of the average allocations and the optimal allocations of selected agents, and the evolution of partial derivatives of cost functions of all agents for a single simulation with different amounts of Gaussian noise added to the partial derivatives of cost functions of agents: 
	(a) with mean zero and standard deviations $\sigma_1 = 50$, $\sigma_2 = 70$ and (b) with mean zero and standard deviations $\sigma_1 = 70$, $\sigma_2 = 110$.} \label{fig3} 
\end{figure}
\begin{figure}[!ht]
	\centering
\subfloat[]{%
			\includegraphics[width=1\linewidth]{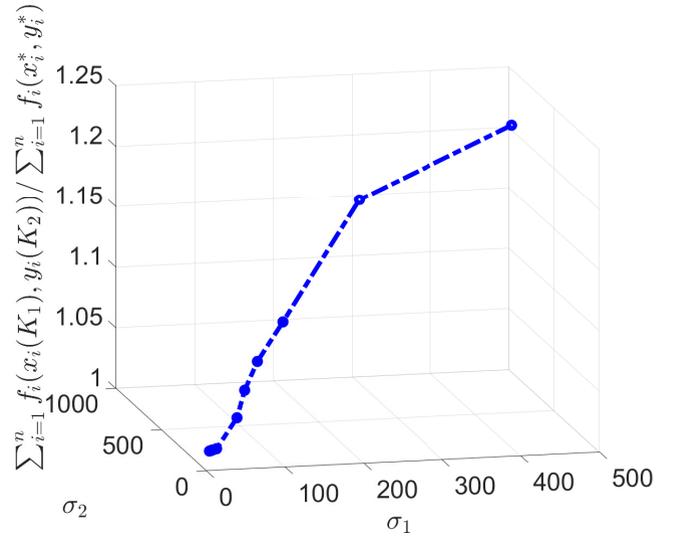}}\\
\subfloat[]{%
		\includegraphics[width=1\linewidth]{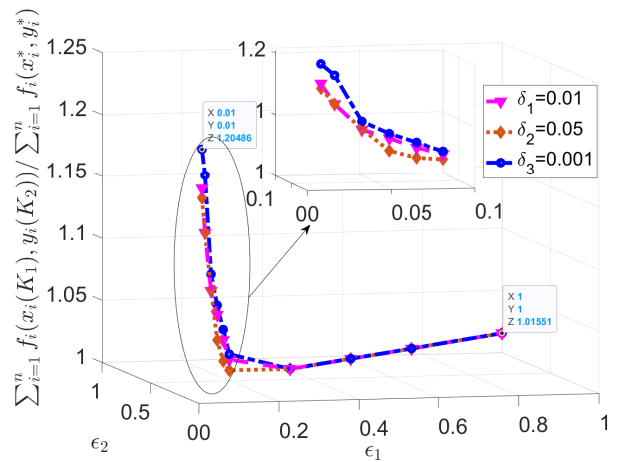}} \hfill
	\caption{The ratio of the total cost obtained by the LDP-AIMD algorithm at the last capacity events and the total optimal cost, (a) evaluated with different standard deviations $\sigma_1$ and $\sigma_2$, (b) evaluated at different values of privacy parameters $\epsilon_1$ and $\epsilon_2$. Here, $K_1$ denotes the experiment's last capacity event for resource $1$, and  $K_2$ denotes the last capacity event for resource $2$.}
	\label{fig4} 
\end{figure}

%

Figure \ref{fig3} illustrates the evolution of the average allocation of resources and the evolution of shaded error bars of partial derivatives of the cost functions of agents with added Gaussian noise of different standard deviations. Note that the larger the standard deviation, the larger the Gaussian noise. Moreover, the absolute difference between long-term average allocations and optimal allocations increases when the standard deviation increases. Additionally, the errors in the shaded error bars of partial derivatives of the cost functions also increase. Thus, the long-term average allocations go farther from the optimal values as noise increases.

Over time the difference between the consecutive partial derivatives of the cost functions decreases; in Figure \ref{fig:l2-sensiv}, we observe that the $l_2$-sensitivities of resources converge over time. 
Moreover, Figure \ref{fig4}(a) shows the ratio of the total cost obtained by the LDP-AIMD algorithm at the last capacity event of the experiment and the total optimal cost evaluated with different standard deviations for several experiments. When the standard deviation increases, thus the noise, then the ratio of the total costs increases. Hence, privacy increases, but accuracy decreases.

The dependence on privacy parameters $\epsilon_1, \epsilon_2$ and $\delta_1, \delta_2$ and the ratio of the total cost obtained by the algorithm at the last capacity events of the experiment and the total optimal cost are illustrated in Figure \ref{fig4}(b). We observe that keeping $\delta_1$ fixed when $\epsilon_1$ is small (better privacy) (analogously, for $\delta_2$ and $\epsilon_2$), the ratio of the costs is larger; hence, lesser accuracy is provided. Furthermore, after a particular value of $\epsilon_1$ and $\epsilon_2$, the ratio of the costs does not change much when the value of $\epsilon_1$ and $\epsilon_2$ increase. We observe that the values of $\epsilon_1$ and $\epsilon_2$ greater than $0.1$ change the ratio of the costs very little. Additionally, similar to $\epsilon_1$ and $\epsilon_2$, when the values of $\delta_1$ and $\delta_2$ are small, the ratio of the costs is large. However, after a particular threshold, the ratio of the costs does not change much when the values of $\delta_1, \delta_2$ increase. 
\begin{figure}[!ht]
	\centering
	\subfloat[]{%
		\includegraphics[width=0.495\linewidth]{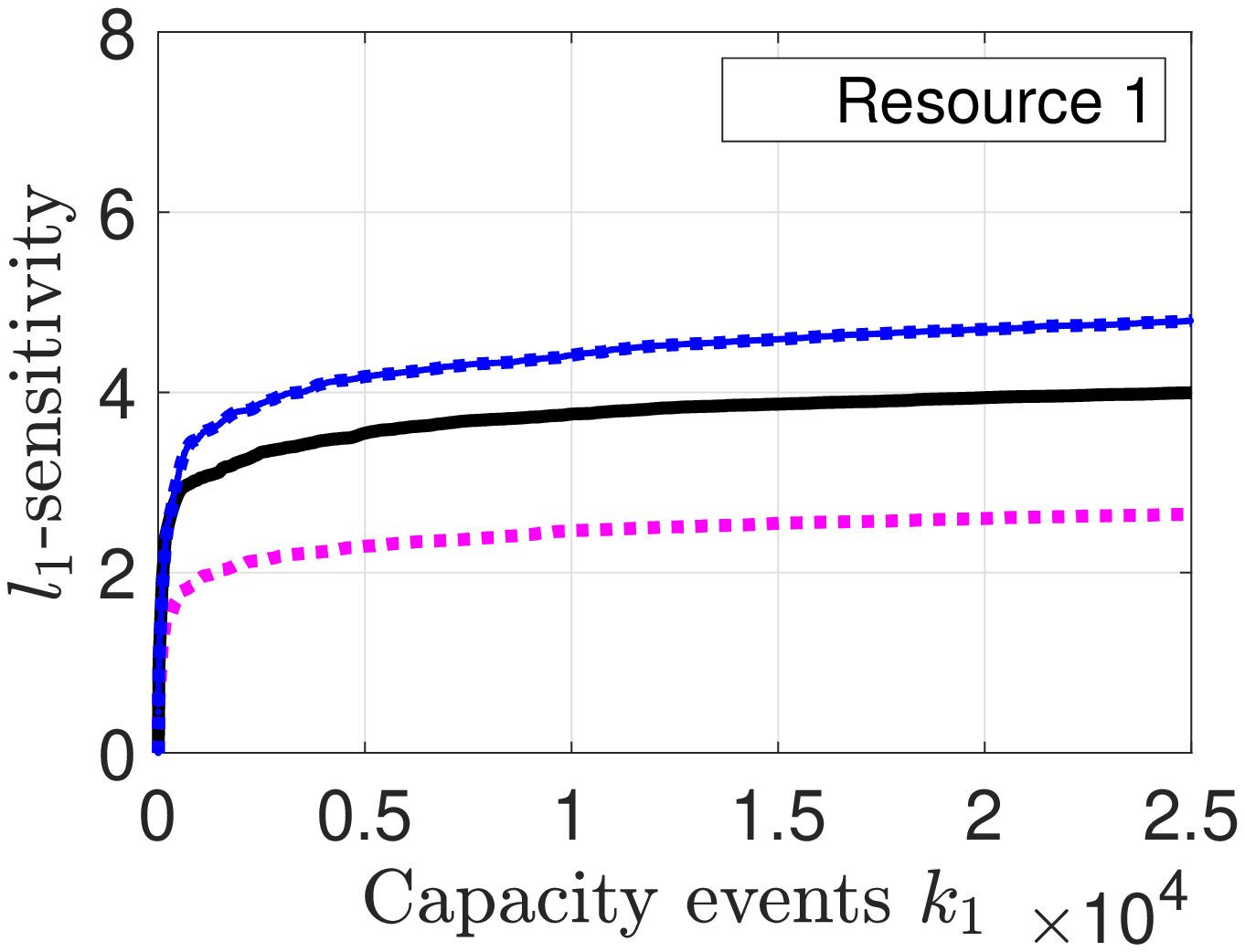}}
	\label{gradx}\hfill
	\subfloat[]{%
		\includegraphics[width=0.495\linewidth]{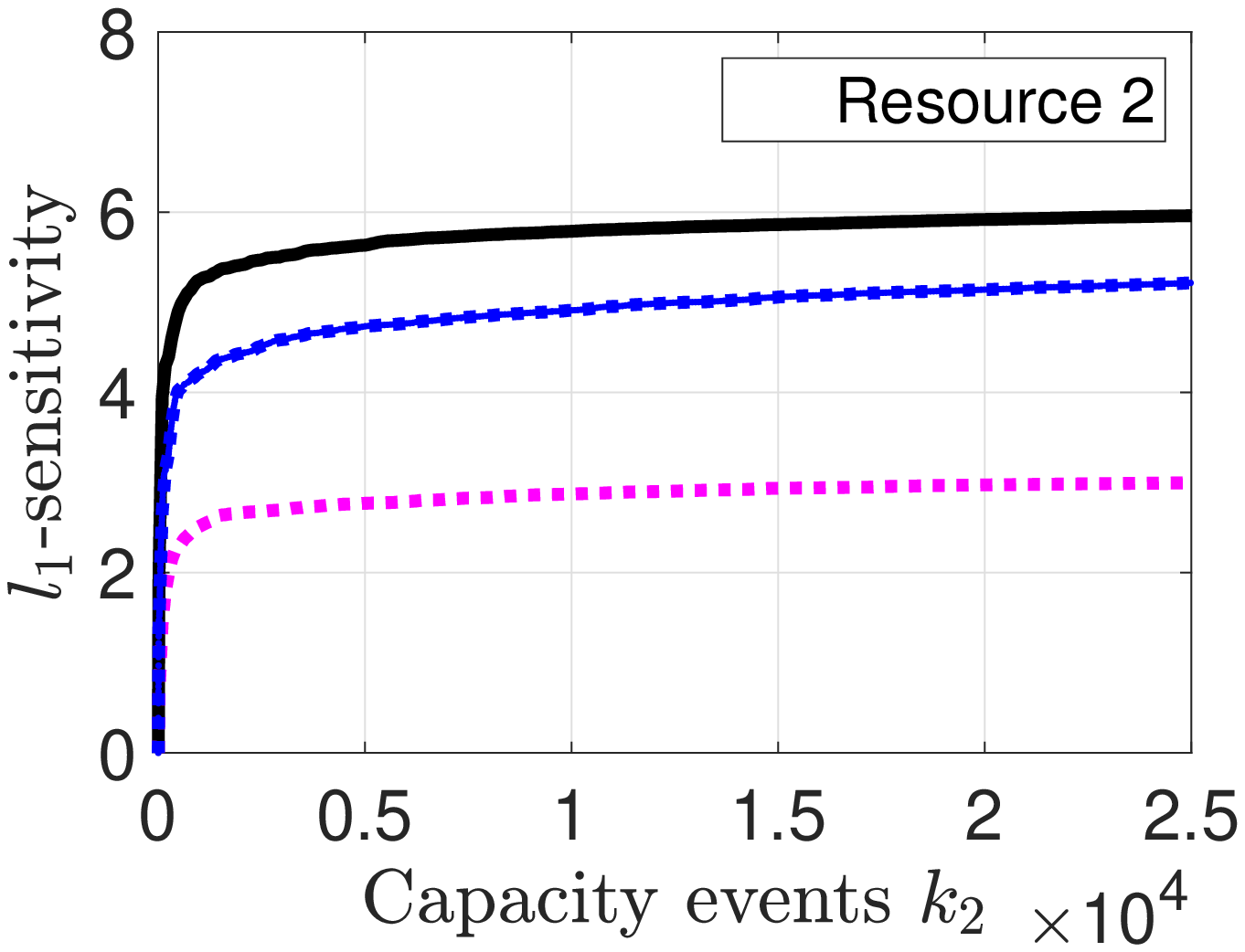}}
	\label{grady} \hfill
	
\caption{(a) Evolution of $l_1$-sensitivity of resource $1$, (b) evolution of $l_1$-sensitivity of resource $2$.} \label{fig:l1-sensiv} 
\end{figure}

\begin{figure}[!ht]
	\centering
	\includegraphics[width=1.1\linewidth]{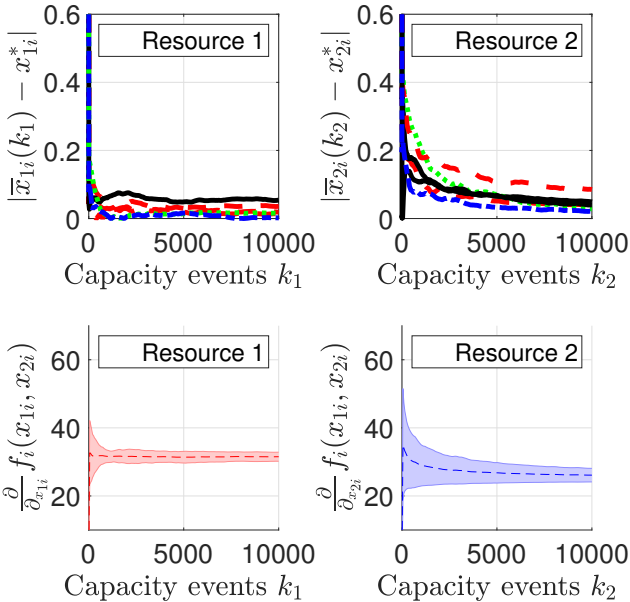}
	\label{gradx}\hfill
\caption{Evolution of absolute difference of the average allocations and the optimal allocations of selected agents, and the evolution of partial derivatives of cost functions with Laplacian noise, $l_1$-sensitivities $\Delta_1 = 5.9$ and $\Delta_2 =  6.34$ and privacy budgets $\epsilon_1= \epsilon_2 = 0.1$.} \label{fig:Lap_abs-diff} 
\end{figure}

\subsubsection{\bf{With Laplacian noise:}}
This subsection presents the results for $(\epsilon_{1i}+\epsilon_{2i})$-local differential privacy (See Theorem 2, Section 4) obtained by adding Laplacian noise to the partial derivatives of the cost functions of agents. In this case, the $l_1$-norm of the consecutive partial derivatives of cost functions of agents is calculated.
Over time the difference between the consecutive partial derivatives of the cost functions decreases as illustrated in Figure \ref{fig:l1-sensiv}; that is, the $l_1$-sensitivities of resources converge over time.
Figure \ref{fig:Lap_abs-diff} illustrates the evolution of the average allocation of resources and the evolution of shaded error bars of partial derivatives of the cost functions of agents with added Laplacian noise. Similar to the Gaussian case, the absolute difference between long-term average allocations and optimal allocations increases when the standard deviation increases. Additionally, the errors in the shaded error bars of partial derivatives of the cost functions also increase---the long-term average allocations go farther from the optimal values as noise increases.

\begin{figure}[!ht]
	\centering
		\includegraphics[width=0.9\linewidth]{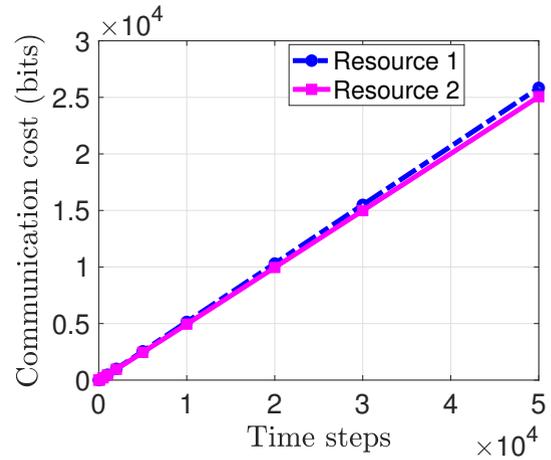}
\hfill	
\caption{{ Evolution of the communication cost (bits) over time steps with Laplacian noise, $l_1$-sensitivities $\Delta_1 = 5.9$ and $\Delta_2 =  6.34$ and privacy budgets $\epsilon_1= \epsilon_2 = 0.1$.}} \label{fig:comm_cost} 
\end{figure}
{ The evolution of the communication cost in bits with the Laplacian noise is illustrated in Figure \ref{fig:comm_cost}; it is the number of capacity events over time steps. We observe that the communication cost increases (approximately) linearly as the time steps increase. It is in line with our observation that the higher the time to convergence, the higher the communication cost.}

\section{Related work} \label{sec:background}
This section briefly presents the related work on distributed resource allocation, distributed optimization, federated optimization, and differential privacy.

\subsection{Resource allocation} 
Divisible resource allocations are studied by Ghodsi et al. \cite{Ghodsi2011}; they propose a dominant resource fairness algorithm for multiple divisible resource allocations. Recently, Nguyen et al. \cite{Nguyen2019} proposed a market-based algorithm to allocate multiple resources in fog computing. Furthermore, Fossati et al. \cite{Fossati2020} proposed a multi-resource allocation for network slicing in the context of {5G} networks. The distributed resource allocation for {5G}-vehicle-to-vehicle communication was proposed by Alperen et al. \cite{Alperen2020}, and an envy-free fair allocation mechanism is proposed in \cite{Ashari2019} for allocating multiple computing resources of servers. { Additionally, interested readers can refer to Shang's work \cite{YShang2019} on a distributed graph-theoretic opinion formation model for social dynamical systems. The model provides resilience to malicious and Byzantine actors in the network.} A survey on the allocation of multiple resources can be found at \cite{Poullie2018}.   
{ An interesting work by Zhou et al. on non-linear dynamic switching heterogeneous multiagent systems is found at \cite{Zou2022}. They propose a leaderless approach to obtain consensus on position and velocity states. Their approach is communication efficient and does not require sharing auxiliary dynamic states to obtain consensus; however, some agent interactions are required. Nonetheless, the major difference between our work and this work \cite{Zou2022} is that \cite{Zou2022}'s approach is leaderless. In contrast, we consider a central server that keeps track of the aggregate consumption of resources and sends notifications when the capacity constraint is reached. Furthermore, our approach deals with differential privacy guarantees.}

\subsection{Distributed optimization} \label{sec:Dist-opt} 

Tsitsiklis and co-authors proposed the distributed optimization problem in their seminal work \cite{Tsitsiklis1986, Tsitsiklis1984}.
To solve a distributed optimization problem, agents in a network need to share their states or multipliers with at least one neighbor, which may compromise the agents' privacy.
Some of the consensus-based distributed optimization approaches are:  Sub-gradient methods \cite{Nedic2009, Romao2021}, Gossip algorithm for distributed averaging \cite{Boyd2006_2, Pu2021}, Dual averaging \cite{Duchi2012, Han2021}, Broadcast-based approach \cite{Nedic2011, Silvestre2019}, to name a few.
Moreover, privacy-preserving distributed optimization has been studied in \cite{Han2021, Huo2022}.
\subsection{Federated optimization} \label{sec:federated_opt}
{\em Federated learning} is a distributed machine learning technique in which several agents (clients) collaborate to train a global model without sharing their local on-device data. Each agent updates the global model with its local data-set and parameters and shares the updates with the central server. The central server aggregates the updates by agents and updates the global model \cite{Mcmahan2017, Konecny2016, Kairouz2021}. 
The optimization techniques are called {\em federated optimization}  \cite{Reddi2021, KonecnyMR2015}.

One of the most popular and widely used federated optimization techniques is FederatedAveraging (FedAvg) by McMahan and co-authors \cite{Mcmahan2017}. It is based on stochastic gradient descent. In the algorithm, a fixed number of agents are selected randomly from agents in the network at each iteration. Each selected agent performs the stochastic gradient descent on its local data using the global model for a certain number of epochs. Moreover, each agent performs the local updates for the same number of epochs. Finally, it communicates the averaged gradients to the central server. The central server then takes a weighted average of the gradients by agents and updates the global model. The process is repeated until the model is trained.

Another federated optimization technique is FedProx \cite{Li2020}, which is a generalization of the FedAvg. In FedProx, the number of epochs is not fixed for each iteration, as in FedAvg; however, it varies, and partial updates by agents are averaged to update the global model. It is proposed for heterogeneous devices and data-sets, data that are not independent and identically distributed (non-IID). In addition, the authors provide convergence guarantees. Sattler and co-authors \cite{Sattler2020} also proposed a federated optimization technique for non-IID data. We list a few other federated optimization techniques such as FedNova \cite{Wang2020}, SCAFFOLD \cite{karimireddy2020}, Overlap-FedAvg \cite{Zhou2022}, federated composite optimization \cite{Yuan2021}, federated adaptive optimization \cite{Reddi2021}. Besides, interested readers can refer to \cite{Kairouz2021}, \cite{Yin2021} and the papers cited therein for a detailed discussion on advances and future research directions in federated optimization and federated learning.

\subsection{Differential privacy}
Han et al. \cite{Han2017} developed a differentially private distributed algorithm to allocate divisible resources. They consider the convex and Lipschitz continuously differentiable cost functions. Moreover, therein noise is added to the constraints of the optimization problem. In another work on differentially private distributed algorithms, Huang et al. \cite{Huang2015} also use convex functions. However, therein noise is added to the cost functions of agents. Olivier et al. \cite{Olivier2020} also developed distributed differentially private algorithms for the optimal allocation of resources. Differential privacy in the context of dynamical systems is studied in \cite{LeNy2014}, \cite{Katewa2018}, and \cite{Ding2018}. Interested readers may refer to the recent books by Le Ny \cite{LeNy2020-book} and Farokhi \cite{Farokhi2020} for further details.
Fioretto et al. \cite{Fioretto2020} developed a differentially private mechanism for Stackelberg games. Duchi et al. developed the local differential privacy mechanism \cite{Duchi2013}. Furthermore, Dobbe et al. \cite{Dobbe2020} proposed a local differential privacy mechanism to solve convex distributed optimization problems. Also, Chen et al. \cite{Chen2021} proposed the local and the shuffle differentially private models.

\section{Conclusion}  \label{conc}
We developed a local differentially private additive-increase and multiplicative-decrease multi-resource allocation algorithm for federated settings wherein inter-agent communication is not required. Nevertheless, a central server aggregates the resource demands by all agents and sends a one-bit notification in the network when a capacity constraint is violated. The algorithm incurs little communication overhead. Furthermore, agents in the network receive close to optimal values with differential privacy guarantees to agents' partial derivatives of the cost functions. We presented numerical results to show the efficacy of the developed algorithms. Also, an analysis is presented on different privacy parameters and the algorithm's accuracy. { The communication complexity of the algorithm is independent of the number of agents in the network. For $m$ shared resources in the network, in the worst-case scenario, $m$ bits are required to communicate the capacity constraint notification.} The work can be extended to other privacy models, for example, the shuffle model. { It will be interesting to show the theoretical results on convergence and the rate of convergence of the LDP-AIMD algorithm.} It will also be interesting to implement the developed algorithm in real-world applications.

 \appendix

 \section{Proofs}
 \label{proofs}
 
 \begin{proof}[Proof of Theorem \ref{thm:Laplacian_loc}]
Let until time instant $t_{k}$, $k_1$ capacity events of resource $1$ and $k_2$ capacity events of resource $2$ occur, respectively. Recall that $f^o_{1,i}(t_{k})$ denotes $\frac{\partial}{\partial{x}} \big \vert_{x=\overline{x}_{1i}(t_{k})} f_{i}(\overline{x}_{1i}(t_{k}), \overline{x}_{2i}(t_{k}))$ and $f'^o_{1,i}(t_{k})$ denotes $\frac{\partial}{\partial{x}} \big \vert_{x=\overline{x}'_{1i}(t_{k})} f_{i}(\overline{x}'_{1i}(t_{k}), \overline{x}'_{2i}(t_{k}))$. Furthermore,  $f^o_{2,i}(t_{k})$ denotes $\frac{\partial}{\partial{x}} \big \vert_{x=\overline{x}_{2i}(t_{k})} f_{i}(\overline{x}_{1i}(t_{k}), \overline{x}_{2i}(t_{k}))$

and $f'^o_{2,i}(t_{k})$ denotes $\frac{\partial}{\partial{x}} \big \vert_{x=\overline{x}'_{2i}(t_{k})} f_{i}(\overline{x}'_{1i}(t_{k}), \overline{x}'_{2i}(t_{k}))$. We obtain the following for resource $1$:
\begin{align}
&\frac{\mathbb{P} \left(M_{q_{1i}} \left (f^o_{1,i}(t_{k})) \right) = \eta_{1i}(t_{k}) \right) }{\mathbb{P} \left(M_{q_{1i}} \left (f'^o_{1,i}(t_{k}) \right) = \eta_{1i}(t_{k}) \right)}  \nonumber \\&= \frac{\mathbb{P} \left(f^o_{1,i}(t_{k}) + d_{1i}(t_{k}) = \eta_{1i}(t_{k}) \right) }{\mathbb{P} \left (f'^o_{1,i}(t_{k}) + d_{1i}(t_{k}) = \eta_{1i}(t_{k}) \right)}. 
\end{align}
Let $k_1$ capacity events for resource $1$ and $k_2$ capacity events for resource $2$ occur until time instant $t_k$. As the noise is drawn from the Laplace distribution and $\norm{f^o_{1,i}(t_{k}) - f'^o_{1,i}(t_{k})}_1$ denotes the $l_1$-sensitivity $\Delta {q_{1i}}(t_k)$, we obtain
\begin{align*}
 &\frac{\mathbb{P} \left(d_{1i}(t_{k}) = \eta_{1i}(t_{k}) - f^o_{1,i}(t_{k}) \right)}{\mathbb{P} \left (d_{1i}(t_{k}) = \eta_{1i}(t_{k}) - f'^o_{1,i}(t_{k}) \right)} \\&= \frac{\exp \left( \frac{\norm{\eta_{1i}(t_{k}) - f^o_{1,i}(t_{k})}_1}{ (\Delta {q_{1i}(t_k)}/\epsilon_{1i})} \right)/ (2\Delta {q_{1i}(t_k)}/\epsilon_{1i})}{\exp \left (\frac{\norm{\eta_{1i}(t_{k}) - f'^o_{1,i}(t_{k})}_1}{ (\Delta {q_{1i}(t_k)}/\epsilon_{1i}} \right) / (2\Delta {q_{1i}(t_k)}/\epsilon_{1i})} \\
 \\&= \frac{\exp \left( \epsilon_{1i}  \frac{\norm{\eta_{1i}(t_{k}) - f^o_{1,i}(t_{k})}_1}{ (\Delta {q_{1i}(t_k)})} \right)}{\exp \left (\epsilon_{1i} \frac{\norm{\eta_{1i}(t_{k}) - f'^o_{1,i}(t_{k})}_1}{(\Delta {q_{1i}(t_k)}} \right)} \\
	& \leq \exp \left( \epsilon_{1i} \frac{\norm{f^o_{1,i}(t_{k}) - f'^o_{1,i}(t_{k})}_1}{\Delta {q_{1i}(t_k)}} \right)\\
	& = \exp \left(\frac{\epsilon_{1i} \Delta q_{1i}({t_k})}{ \Delta {q_{1i}(t_k)}} \right) = \exp(\epsilon_{1i}).
\end{align*}

Similarly, we obtain for resource $2$:
\begin{align}
\frac{\mathbb{P} \left(M_{q_{2i}} \left (f^o_{2,i}(t_{k}) \right) = \eta_{2i}(t_{k}) \right) }{\mathbb{P} \left(M_{q_{2i}} \left (f'^o_{2,i}(t_{k}) \right) = \eta_{2i}(t_{k}) \right)} \leq \exp(\epsilon_{2i}).
\end{align}

We obtained the following result for the coupled privacy mechanism:
\begin{align*}
&\frac{\mathbb{P} \left(\mathcal{A}_{i} \left (f^o_{1,i}(t_{k}), f^o_{2,i}(t_{k}) \right) = (\eta_{1i}(t_{k}), \eta_{2i}(t_{k}))  \right) }{\mathbb{P} \left(\mathcal{A}_{i} \left (f'^o_{1,i}(t_{k}),f'^o_{2,i}(t_{k}) \right) =  ( \eta_{1i}(t_{k}), \eta_{2i}(t_{k})) \right)} \nonumber\\
&=  \frac{\mathbb{P} \left(M_{q_{1i}} \left (f^o_{1,i}(t_{k}) \right) = \eta_{1i}(t_{k})  \right) }{\mathbb{P} \left(M_{q_{1i}} \left (f'^o_{1,i}(t_{k}) \right) = \eta_{1i}(t_{k}) \right)} \\& \times \frac{\mathbb{P} \left(M_{q_{2i}} \left (f^o_{2,i}(t_{k}) \right) = \eta_{2i}(t_{k})  \right) }{\mathbb{P} \left(M_{q_{2i}} \left (f'^o_{2,i}(t_{k}) \right) = \eta_{2i}(t_{k}) \right)} 
\nonumber \\
& \leq \exp(\epsilon_{1i})\times \exp(\epsilon_{2i})
= \exp(\epsilon_{1i} + \epsilon_{2i}). 
\end{align*}

\end{proof}

\begin{proof}[Proof of Theorem \ref{thm:Gaussian_loc}]
For fixed $j$ and $i$, the proof for $(\epsilon_{ji}, \delta_{ji})$-differential privacy is similar to the proof of Theorem A.1 of \cite{Dwork2014}. We then use joint probability for multiple resources and follow steps similar to the proof of Theorem 4.1. 
\end{proof}


\bibliographystyle{IEEEtran}
\bibliography{DP_AIMDBib}

\begin{IEEEbiography}[{\includegraphics[width=1in,height=1.25in,clip,keepaspectratio]{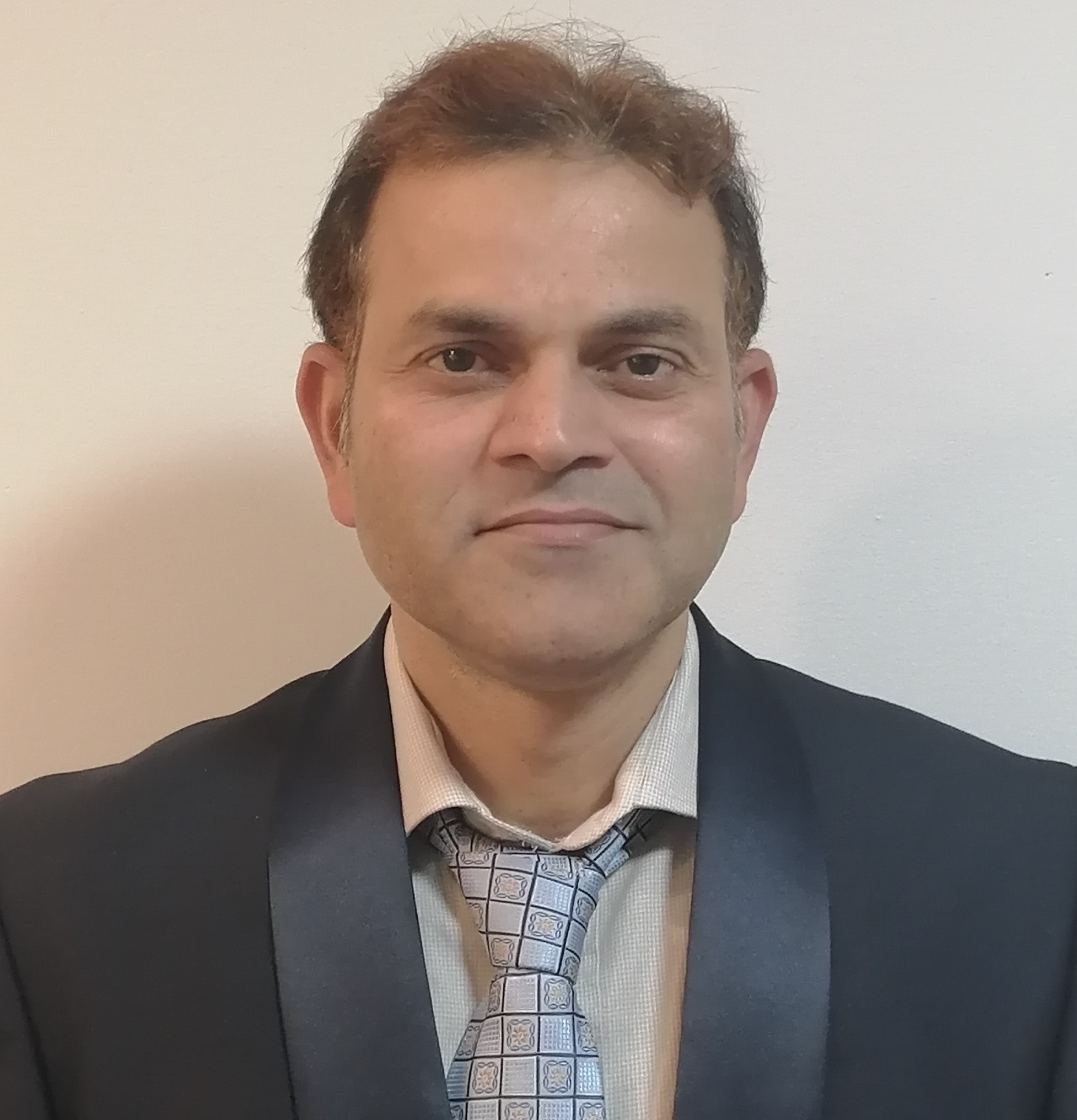}}]{Syed Eqbal Alam} is a Postdoctoral fellow at the J. Herbert Smith Center for Technology Management \& Entrepreneurship, Faculty of Engineering, University of New Brunswick, Fredericton, New Brunswick, Canada. He obtained his Ph.D. from the Concordia Institute for Information Systems Engineering (CIISE), Concordia University, Montreal, Quebec, Canada. Before enrolling in the Ph.D. program, he was a lecturer at Taif University, Saudi Arabia. He taught undergraduate computer science courses such as object-oriented programming, data structure and algorithms, distributed systems, artificial intelligence, and discrete mathematics. Furthermore, Syed Eqbal obtained his Master of Technology degree from the International Institute of Information Technology, Bangalore (IIIT-B), India.

 Dr. Syed Eqbal's interest is in performing research that leads to excellent theories and exciting applications. He broadly works on distributed optimization, federated optimization, differential privacy, Internet-of-things, federated learning, and algebraic structures. 
  Dr. Syed Eqbal has been a reviewer for several conferences and journals, including IEEE Conference on Decision and Control, European Conference on Artificial Intelligence, IEEE Transactions on Automatic Control, and IEEE Access.
\end{IEEEbiography}
\begin{IEEEbiography}[{\includegraphics[width=1in,height=1.25in,clip,keepaspectratio]{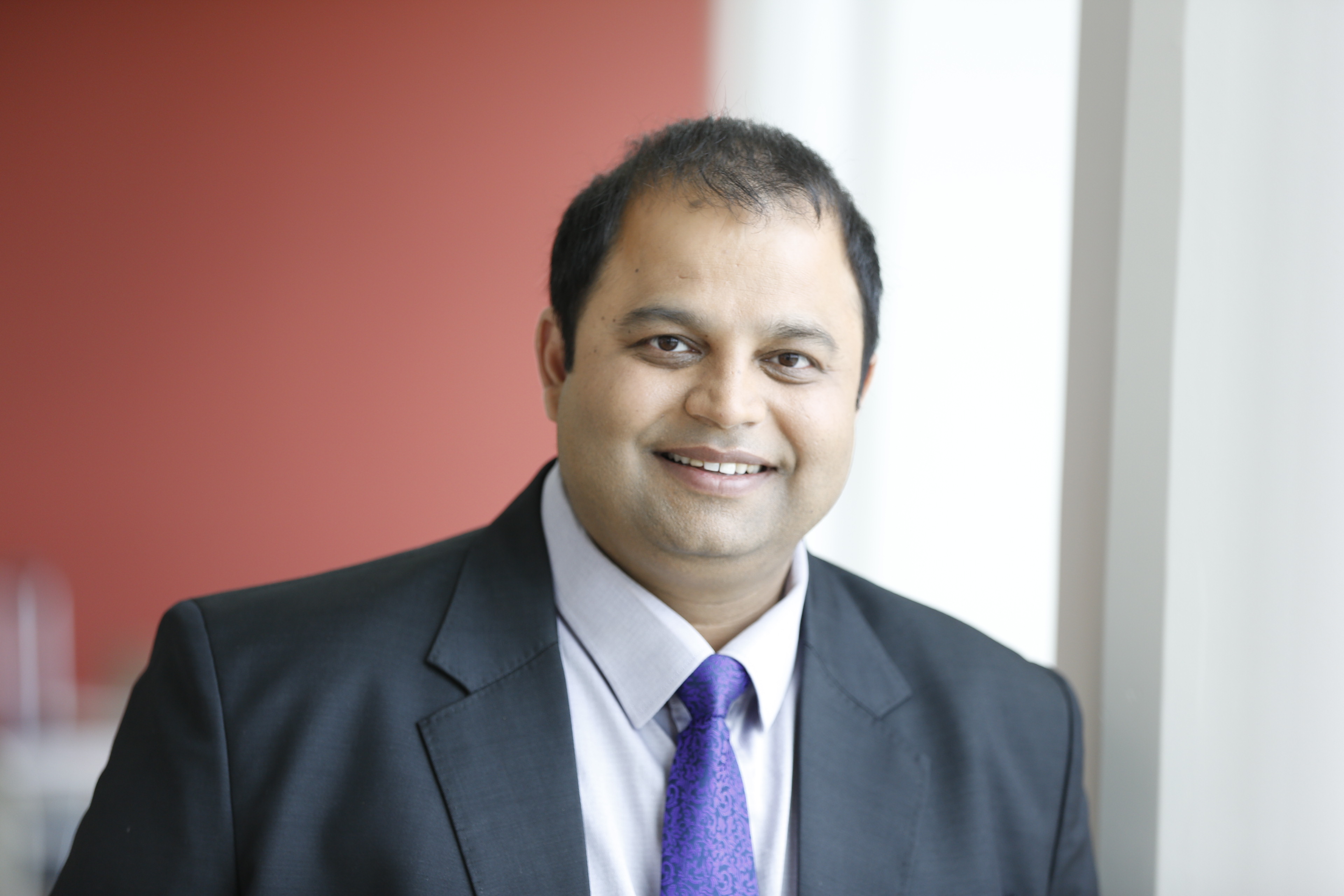}}]{Dhirendra Shukla} obtained his Ph.D. in Entrepreneurial Finance from the University of London in the UK (King's College London). He obtained his MBA from the Telfer School of Management at the University of Ottawa (Canada). Dr. Dhirendra Shukla is a Professor and the Dr. J Herbert Smith ACOA Chair in Technology Management and Entrepreneurship 
at the University of New Brunswick, Fredericton. 

Dr. Dhirendra utilizes his expertise from the telecommunication sector and extensive academic background in the areas of entrepreneurial finance, business administration, and engineering to promote a bright future for New Brunswick, Canada. Passionate about entrepreneurship, design, engineering, innovation, and leadership, Dr. Dhirendra ignites this same passion in others through his teaching at UNB, local entrepreneurial community outreach, and the creation of programs such as the Master of Technology Management and Entrepreneurship, Technology Commercialization Program and Summer Institute. Recognition of his tireless efforts and vision is demonstrated through UNB’s 2014 award from Startup 
Canada as “Most Entrepreneurial Post Secondary Institution of the Year”,  his nomination as a finalist for Industry Champion by KIRA, and most recently, his nomination as a finalist for Progress Media’s Innovation in Practice award. 
\end{IEEEbiography}
\begin{IEEEbiography}[{\includegraphics[width=1in,height=1.25in,clip,keepaspectratio]{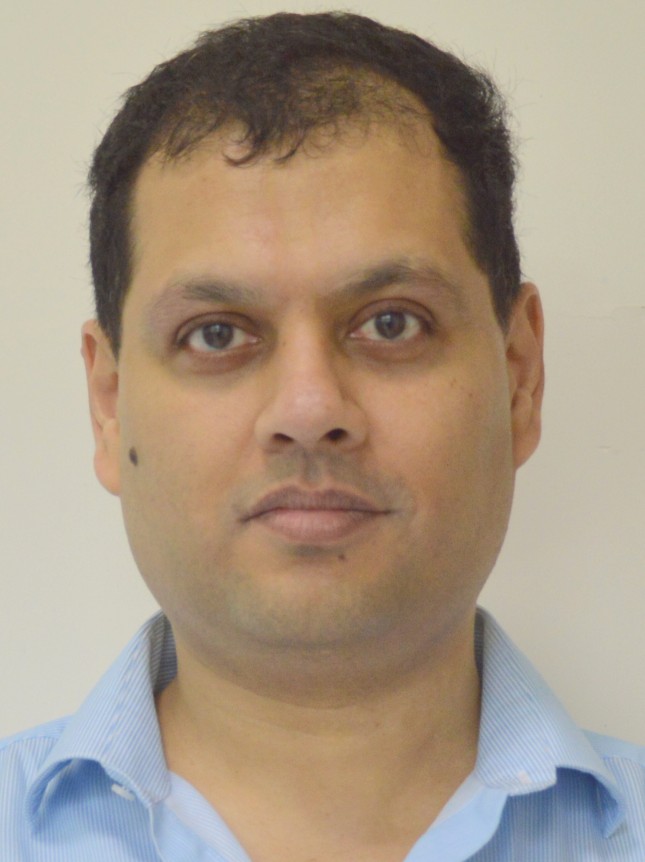}}]{Shrisha Rao} received his Ph.D. in computer science from the University of  Iowa, and before that his M.S. in logic and computation from Carnegie 
Mellon University.  He is a professor at IIIT-Bangalore, a graduate school 
of information technology in Bangalore, India.  His primary research 
interests are in applications of agent-based modeling and artificial 
intelligence to understand human cognition and emergent social phenomena.  
He also has interests in sustainability, bioinformatics, and intelligent 
transportation systems.  

Dr. Rao was an ACM Distinguished Speaker from 2015 to 2021, and is a 
Senior Member of the IEEE.  He is also a life member of the American 
Mathematical Society and the Computer Society of India.

Dr. Rao is on the editorial boards of Sādhanā, The Knowledge Engineering 
Review, and PLOS ONE, and has had editing assignments for the IEEE Systems 
Journal and Frontiers in Public Health.   He has been a reviewer for many 
journals of several major publishers, and has served as a program 
committee member for several international conferences. 
\end{IEEEbiography}

\EOD

\end{document}